\documentclass[12pt]{article}
\pdfoutput=1
\usepackage[cmex10]{amsmath}
\usepackage{amsfonts}%
\usepackage{amssymb}%
\usepackage{bbm}
\usepackage{amsthm}
\usepackage{hyperref}
\usepackage{cleveref}
\usepackage[pdftex]{graphicx}
\usepackage{epstopdf}
\usepackage{psfrag}
\usepackage{tipa}
\usepackage[boxed]{algorithm2e}
\usepackage{array}
\usepackage{mdwmath}
\usepackage{mdwtab}
\usepackage{eqparbox}
\usepackage{subfigure}
\usepackage{fullpage}
\usepackage{cite}
\usepackage{xcolor}
\usepackage{soul}

\newtheorem{theorem}{Theorem}
\newtheorem{lemma}{Lemma}

\newtheorem{proposition}{Proposition}
\newtheorem{corollary}{Corollary}

\newtheorem{claim}{Claim}

\newtheorem{remark}{Remark}

\newcommand{\off}[1]{}

\newcommand{\R}{\mathbb{R}}
\newcommand{\N}{\mathbb{N}}
\newcommand{\C}{\mathbb{C}}

\newcommand{\cI}{{\cal I}}

\newcommand{\E}{\mathrm{E}}
\newcommand{\e}{\mathrm{e}}

\newcommand{\bw}{\mathbf{w}}
\newcommand{\bu}{\mathbf{u}}
\newcommand{\bU}{\mathbf{U}}
\newcommand{\bV}{\mathbf{V}}

\newcommand{\by}{\mathbf{y}}
\newcommand{\bH}{\mathbf{H}}
\newcommand{\bh}{\mathbf{h}}

\newcommand{\bM}{\mathbf{M}}
\newcommand{\bR}{\mathbf{R}}

\newcommand{\bz}{\mathbf{z}}

\newcommand{\bx}{\mathbf{x}}

\newcommand{\bxi}{\pmb{\xi}}
\newcommand{\indicator}[1]{\mathbbm{1}_{\left\{ {#1} \right\} }}

\newcommand{\SEFI}[1]{{\bfseries *** Sefi says: #1 ***}}

\newcommand{\RED}[1]{}
\newcommand{\BLUE}[1]{#1}
\newcommand{\blue}[1]{#1}
\newcommand{\red}[1]{}
\newcommand{\redo}[1]{}

\newcommand{\prnt}[1]{\left( #1 \right)}
\newcommand{\brkt}[1]{\left[ #1 \right]}

\begin{document}
\title{The \BLUE{Ergodic} Capacity of the Multiple Access Channel Under Distributed Scheduling - Order Optimality of Linear Receivers\thanks{Parts of this work appeared in the proceedings of ITW 2013, Seville, Spain.}}
\author{
Joseph Kampeas, Asaf Cohen and Omer Gurewitz
\\
Department of Communication Systems Engineering
\\
Ben-Gurion University of the Negev
\\
Beer-Sheva, 84105, Israel
\\ {\tt Email: \{kampeas,coasaf,gurewitz\}@bgu.ac.il}}
\date{}
\maketitle
\begin{abstract}
Consider the problem of a Multiple-Input Multiple-Output (MIMO)
Multiple-Access Channel (MAC) at the limit of large number of users.
Clearly, in practical scenarios, only a small subset of the users can
be scheduled to utilize the channel simultaneously. Thus, a problem of
user selection arises. However, since solutions which collect Channel
State Information (CSI) from all users and decide on the best subset
to transmit in each slot do not scale when the number of users is
large, distributed algorithms for user selection are advantageous.

In this paper, we analyse a distributed user selection algorithm,
which selects a group of users to transmit without coordinating
between users and without all users sending CSI to the base
station. This threshold-based algorithm is analysed for both Zero-Forcing (ZF) and Minimum Mean Square Error (MMSE) receivers, and its expected  sum-rate in the limit of large number of users is investigated.
It is shown that for large number of users it achieves the same scaling laws as the optimal centralized scheme.
\end{abstract}
\section{Introduction}\label{sec. intro}
Wireless access networks are the typical last-mile networks connecting multiple users to a high speed backbone. In these networks, a Base Station (BS) serves a large group of users. Traditionally, either the time or the frequency are divided to ensure users do not interfere with each other. Modern coding techniques, however, allow multiple users to either transmit or receive simultaneously, and be decoded successfully using appropriate Multiple Access Channel (MAC) codes or Broadcast Channel (BC) codes, respectively.

In prevalent scenarios in which the BS is serving a large user population it is not practical to serve all users simultaneously, and the BS is required to select a subset of users to be served. In such scenarios the problem of user selection or scheduling is acute, i.e., selecting a proper subset of users to be served can dramatically boost the performance in terms of a predefined objective, e.g., can dramatically increase the achievable sum-rate. As shown in literature, such performance enhancement is applicable both in the case of downlink traffic, in which a BS transmits to a group of users, as well as for the case of upstream traffic, in which a selected sub-set of users are transmitting simultaneously to the BS. Furthermore, numerous papers have shown that under various setups optimal or near-optimal sum-rate can be attained by selecting a proper subset of users, e.g., it was shown that for the downstream traffic, selecting a sub-set of users with good channel condition, yet with relatively orthogonal channels and utilizing zero-forcing beamforming (ZFBF) strategy achieves the performance of Dirty Paper Coding (DPC), hence is asymptotically optimal in the Gaussian case (e.g., \cite{yoo2006optimality,jagannathan2007scheduling,swannack2004low,kim2005scheduling,airy2003spatially}).

The challenge in selecting a proper set of users under the predefined objective is twofold. First, in order to find a proper set the scheduler (e.g., the BS) needs to attain the channel state between each and every BS's antenna and each potential user's antenna or vice versa (between the potential users' antennas and each BS's antennas), for the downstream or upstream traffic, respectively. Obviously, such channel state acquisition consumes large overhead which hinders the gain of serving multiple users simultaneously, especially when the number of users is moderate or large. Second, even after possessing the channel state between each and every potential transmit antenna and each and every potential receive antenna, searching for the optimal user set is computationally prohibitive. Note that most of the literature considering this problem assumes the Channel State Information (CSI) is available at the scheduler and suggests heuristics for user selection procedures for various setups, e.g., in \cite{yoo2006optimality} an efficient user selection heuristic was suggested which, given the CSIs, attains the capacity scaling law when the number of users is large.

In this paper, we  analyse a distributed threshold-based user selection algorithm for the upstream traffic case, in which each user determines for itself whether to transmit or not, in each transmission opportunity, without any coordination between the users and with minimal CSI exchange (specifically, the BS attains CSIs only from the self-selected users). We analyse this threshold-based algorithm both for Zero-Forcing (ZF) and Minimum Mean Square Error (MMSE) detection, \RED{scheme,} and show that in both cases the \RED{expected} \BLUE{ergodic} sum-rate for large number of users achieves the same scaling laws as the one obtained by the optimal centralized scheme. Note that employing such distributed threshold-based mechanisms tackles both challenges mentioned above. First, it requires minimal CSI exchange, and second, it provides a simple and distributed user selection search. 

\subsubsection*{Main Contribution}\label{subsec. contribution}
We consider a MIMO MAC channel with $r$ receiving antennas and $K \gg r$ users. We investigate a distributed algorithm for selecting a group of users to transmit in each slot.  In particular, a threshold value for the norm of the channel vector is set, and only users above the threshold transmit. Hence, there is no need to collect CSIs from all users. Nor is any cooperation required.

The contribution of this study is the analysis of the resulting \BLUE{ergodic} sum-rate in the limit of large $K$. The respective scaling laws are given for both Zero-Forcing (ZF) and Minimum Mean Square Error (MMSE) receivers. This analysis employs recent tools from both Point Process approximation and asymptotic random matrix theory, that to the best of our knowledge, were not used in this setting before. Via this analysis, the simple distributed, threshold based algorithm, is shown to achieve the optimal scaling laws. 

The rest of this paper is organized as follows: Section~\ref{sec. related} presents the most relevant related work. Section~\ref{sec. prelim} includes the required preliminary material. Section~\ref{sec. dist alg} introduces the distributed algorithm, gives its analysis and scaling law under ZF decoding. Section~\ref{sec. MMSE rec} gives the analysis under an MMSE receiver. While conceptually similar, this analysis is the more technically challenging. Section~\ref{sec. conc} concludes the paper. \blue{Proofs which were omitted in the main body are given in Appendix~\ref{sec:proofs}}.

\section{Related Work}\label{sec. related}

The essence of multi-user diversity was introduced in \cite{knopp1995information}, where selecting the strongest user in each time slot was first suggested.
In \cite{gozali2003impact}, the authors considered the impact of multi-user diversity on the MIMO channel. Assuming channel state information at the BS, the authors used order statistics to evaluate the effective SNR when scheduling the strongest user in each slot.
Nevertheless, locating the strongest user is not an easy task, especially for large number of users, as it requires huge overhead due to excessive CSI exchange. A practical method to attain multi-user diversity is to schedule the strongest user distributively by setting a threshold on, e.g., the link capacity or channel gain.

A pioneering study of distributed scheduling was done in \cite{qin2003exploiting,qin2006distributed}, where a decentralized MAC protocol for Orthogonal Frequency Division Multiple Access (OFDMA) channels was suggested. In this scheme, each user estimated the channel gain and compared it to a threshold. Only above-the-threshold users could transmit.
In \cite{gesbert2004much}, the authors proposed threshold-based scheduling for the SISO downlink channel, and showed that using such a scheme  reduces the required feedback, while preserving most of the system capacity and outage probability.
Recently, we proposed a Point Process approximation which facilitates the analysis of various distributed threshold-based  single user scheduling algorithms in the non-homogeneous scenario \cite{kampeas2012opportunistic,kampeas2012capdisthetnet}.
In \cite{bai2006opportunistic} the authors extended the threshold-based scheme to a multi-channel setup, where each user competes on $m$ channels.
In \cite{qin2008distributed}, the authors used a similar approach for power allocation in the multi-channel setup, and suggested an algorithm that asymptotically achieves the optimal water filling solution.

Yet, to fulfill the full potential of multiple antenna systems, multi-user scheduling should be  considered. However, when extending the single user scheduling problem to a multi-user one, the problem of mutual interference arises. Thus, the scheduled group may have a significant impact on the system performance.
In \cite{kim2005scheduling}, the authors suggested a greedy algorithm to schedule the strongest and most orthogonal users for the MIMO downlink model with Zero-Forcing Beamforming (ZFBF) detector.
A closely related scheme was suggested in \cite{shen2006low}. Using Block Diagonalization (BD), a capacity-based greedy algorithm was suggested, in which first the strongest user is scheduled, and then additional users are added, one by one, based on their marginal contribution to the total capacity.
In the same context, \cite{jagannathan2006efficient} considered the special case of two transmit antennas and one receive antenna per user, and showed that a greedy, two-stage algorithm, which first selects the strongest user and then the second to form the best pair is asymptotically optimal.
In the context of heterogeneous users, \cite{jagannathan2007scheduling} proposed a scheduling scheme which selects a small subset of the users with favorable channel characteristics.
In fact, in the downlink scenario, it was shown that ZFBF with sub-optimal user selection can indeed achieve the Dirty Paper Coding (DPC) region \cite{yoo2006optimality}, and is hence optimal in the Gaussian case \cite{weingarten2006capacity}.
Additional centralized multi-user scheduling algorithms for MIMO communication were given in \cite{airy2003spatially,swannack2004low,primolevo2005channel,yoo2006finite}.
Space-time coding for fading multi-antenna MAC was considered in \cite{gartner2006multiuser}. The focus therein, however, was on joint code design for a given point in the rate region and the resulting error probability, rather than user scheduling and its resulting capacity.
In \cite{pun2007opportunistic}, the authors used Extreme Value Theory (EVT) to derive the scaling laws for scheduling systems using beamforming and linear combining.
In \cite{wang2007coverage}, the authors suggested sub-carrier assignment algorithms that achieve fairness and better cell coverage, and used order statistics to derive an expression for the resulting link outage probability.
An asymptotically optimal scheme for multiple base stations (with joint optimization) was given in \cite{zakhour2011minmax}. \cite{choi2008capacity} analysed the scaling laws of base station scheduling, and showed that by scheduling the strongest among $K$ stations one can gain a factor of $O(\sqrt{2\log K})$ in the expected capacity (compared to random or Round-Robin scheduling).
Additional surveys and scaling laws can be found in \cite{caire2006mimo,hassibi2007fundamental,sharif2007comparison}.

Still, attaining multi-user diversity, and further, selecting a favorable group of users, requires CSI and complex coordination. Thus, a distributed solution is desirable in this case as  well.
In \cite{bayesteh2008user}, the asymptotic performance of a threshold-based algorithm for scheduling users in a MIMO broadcast channel environment under ZF detection was analysed. In particular, the transmitter utilized threshold values on the eigenvalues of the users' channel matrix. Then, among the relevant candidates, a set of nearly orthogonal users were selected. Such  a scheduling scheme in nearly optimal, as the transmitter selects both strong and nearly orthogonal users. Note that to compute a threshold on the eigenvalues, one should deal with the eigenvalue distribution, which does not have a closed form expression. In this paper, on the other hand, we compute a threshold on the channel norm, which has a simple form of the Chi-square distribution, and thus, the analysis herein is inherently different.
 In \cite{sharif2005capacity}, the authors suggested a random beaming scheme for the broadcast channel, in which the transmitter \emph{first chooses the directions} at random, then selects users that obtained the highest SINR values in the directions chosen. Such a scheme requires very limited feedback, nevertheless, still achieves the optimal scaling law. However, random beaming schemes suffer from degraded power gains, as a receiver (in an uplink model) sees only the projections on the a-priori chosen directions, and cannot use the full knowledge on the true directions of the users actually transmitting. This can be critical, especially in the low SINR regime. In this paper, however, the receiver can detect the transmitting users based on their exact directions, since it has the transmitting users' CSI and no a-priori directions were defined. Hence, the resulting power gain and its analysis are different. In other words, herein, we are able to design the detection matrices based on the channel vectors of the selected users, and we are not restricted to a (random) set of directions chosen in advanced.

A key contribution of the current work is the non-trivial extension of the work in \cite{kampeas2012opportunistic} to analyse multiple-access protocols, where several users transmit simultaneously and should be decoded successfully, hence the questions that arise are: $(i)$ how to distributively select a good subset of users to transmit, $(ii)$ what is the mutual influence between the users in the selected group, and $(iii)$ what is the performance under different detection schemes.

\BLUE{Channel hardening refers to the phenomenon that the variance of the limiting distribution of the channel mutual information decreases in relation to the mean as the number of antennas grows \cite{hochwald2004multiple}. Accordingly, if matrix $\bH$ dimensions between the transmit antennas and selected users are large, channel hardening limits the gains provided by scheduling, i.e., the user selection gain of choosing preferable users is much smaller with respect to the gain already achieved by the large number of antennas. Note that channel hardening relies on the law of large numbers and is applied to massive MIMO systems. Accordingly, as long as after scheduling the dimensions of $\bH$ remain large, user selection gain is much smaller compared to gain already achieved by the large number of antennas.}

\BLUE{In this work, however, even though the number of potential users is large, the number of actually selected users is small, and since the number of receive antennas is small as well, the dimensions of $\bH$ are fixed, and channel hardening does not apply.}

\BLUE{In fact, one of this study's key contributions is in considering a different kind of asymptotic in which the size of $\bH$ is fixed (small), and taking the \emph{number of users from which to select, to infinity}. Moreover, due to the distributed algorithm we suggest, the entries in $\bH$ are dependent, rendering many previous results useless and requiring us to derive new tools to analyse the \BLUE{sum-rate}.}



%
%
%

\section{Preliminaries}\label{sec. prelim}
In this section, we describe the system model and relevant results which will be used throughout this paper.
\subsection{System Model}

We consider a multiple-access model with $K$ users, each with a single transmit antenna. The BS is equipped with $r$ receiving antennas.
Let us denote random matrices in bold upper-case and random vectors and variables in bold lower-case letters.
When $k$ users utilize the channel simultaneously, the received signal at the base station can be described as:
\begin{equation}
\by = \sum_{i=1}^k \bh_i \bx_i + \bw,
\label{eq: channel model}
\end{equation}
where $\bx_i \in\mathbb{C}$ is the transmitted signal (scalar). $\bx_i$ has a short-term constrain in its total power to $P$, i.e., $\text{E}[\bx_i^\dag \bx_i]\leq P$ per block (users are not allowed to aggregate power if they do not transmit). However, in most cases, we will assume a \emph{constant} power constraint $P$. $\bw\in\mathbb{C}^{r}$ denotes the uncorrelated Gaussian noise. $\bh_i \in\mathbb{C}^{r\times 1}$ is a complex random Gaussian channel vector  whose coefficients describe the gain and phase between the user's transmitting antenna to each of the receiving antennas at the BS. When all users are identically and independently distributed, it is common to assume that all entries of $\bh_i$ are independent and have zero mean and variance $1/2$ imaginary and real parts, for all users.

We assume a memoryless block-fading channel model, i.e., the channel remains constant over each slot (block) period, and at the beginning of each slot independent realizations of $\{\bh_i\}_{i=1}^{K}$ are drawn. Furthermore, we assume that each user has accurate channel gain estimation to each of the BS antennas. 
 \BLUE{Eventually, however, in the algorithm we suggest it is sufficient to know only the norm $\Vert \bh \Vert^2$.} We further assume TDD channel reciprocity, i.e., the downlink and uplink transmissions are performed at the same carrier frequency and uplink transmissions happen at the same coherence time such that downlink channel estimation at the \red{transmitter}\blue{users} can be directly utilized for the uplink transmission \cite{caire1999capacity}.

\BLUE{This paper focuses on a threshold based scheduling for \emph{multiuser selection}. In particular, a group of users schedule themselves if their channels' gains are above a high threshold, without knowing the number of above-threshold users at that slot, or their channel correlation with the other scheduled users. The channel correlation highly affects the resulting rate of each user. That is, even if each user knows its own CSI, due to the inter-user interference, determining the achievable transmission rate is impossible without coordination.  As a result, when each user transmits at a fixed rate $\mathcal{R}$, the BS can decode that transmission only if the capacity is greater than $\mathcal{R}$ at that time. Otherwise, there is an \emph{outage event}~\cite{ozarow1994information}. Alternatively, a user may code over many blocks (time slots) while averaging all fading values and taking into account the inter-user interference statistics, and thus attain the ergodic rate. This paper focuses on the latter approach, i.e., analysing the ergodic capacity under distributed multiuser selection with MMSE and ZF detection.}

\subsection{Multi-User Diversity Via EVT}
EVT is a key tool in \red{evaluating}\blue{our evaluation of} the rate under scheduling in multi-user systems. In this subsection we briefly review the most relevant results in this context. In addition, we develop new normalizing constants for the problem at hand (EVT for the $\chi^2$ distribution), which will later aid at speeding up convergence results.

As the sum rate is mainly influenced by the channel vectors' gains and directions, our goal is to explore this behavior for a large number of users. Specifically, we first wish to explore the behavior of the maximal gain. Since the entries of $\bh$ are complex Gaussian, the channel gain follows a $\chi^2$ distribution with $2r$ degrees of freedom, denoted $\chi^2_{2r}$. \blue{Let `$\stackrel{d}{\to}$' denote convergence in distribution.} We utilize the following EVT theorem.
\begin{theorem}[\cite{de2006extreme,leadbetter1983,embrechts2011modelling}]\label{theorem. EVT}
Let $\bz_1,..\bz_K$ be a sequence of i.i.d.\ random variables with distribution $F(z)$, and let $\bM_K = \max(\bz_1,...,\bz_K)$.
If there exists a sequence of normalizing constants $a_K>0$ and $b_K$ such that as $K\rightarrow \infty$,
\begin{equation*}\label{eqn: P(Mn) approx}
\Pr(\bM_K \leq a_K z + b_K)\stackrel{d}{\to} G(z)
\end{equation*}
for some non-degenerate distribution G, then G is of the generalized extreme value (GEV) distribution type
\begin{equation*}\label{eqn: G def}
G(z) = \exp\left\{-(1+\xi z)^{-1/\xi}\right\}
\end{equation*}
and we say that $F(z)$ is in the domain of attraction of $G$,
where $\xi$ is the shape parameter, determined by the ancestor distribution $F(z)$.
\end{theorem}
\blue{It was shown that} when $\{\bz_K\}$ is a sequence of i.i.d.\ $\chi^{2}_{2r}$ variables, the asymptotic distribution of $\bM_K$ is a Gumbel distribution (e.g., \cite[pp. 156]{embrechts2011modelling}). Specifically,
\begin{equation}
\Pr(\bM_K \leq a_K z + b_K) \stackrel{d}{\to} \e^{-\e^{-z}},
\label{eq: Gumbel distribution}
\end{equation}
where
\begin{eqnarray}\label{eqn: a_n normalized slow}
   a_K &=& 2,
\\
\label{eqn: b_n normalized slow}
  b_K &=&2\left( \log K +(r-1)\log\log K - \log \Gamma(r) \right) + o(1),
\end{eqnarray}
and $\Gamma(r)=\int_{0}^{\infty}t^{r-1}\e^{-t}dt$ is the Gamma function.

\label{commentr22}\blue{However, the convergence of the maxima to the Gumbel distribution is quite slow for i.i.d.\ $\chi^2_{2r}$ random variables, when using the above normalizing constants. That is, the approximation of the maximal value using the above normalizing constants and the Gumbel distribution will not be tight for moderate values of $r,K \in \mathbb{N}$. Thus, to provide insight into practical setups, we devise new normalizing constants for the maxima of $\chi^2$-distribution, which have the same asymptotic limit as the constants in \eqref{eqn: a_n normalized slow} and \eqref{eqn: b_n normalized slow}, yet, as will be shown experimentally, attain an accurate approximation even for moderate $K$ and $r$ (i.e., the approximation using these constants converges much faster). The method used to derive the new constants is technical, and hence deferred to \Cref{proof:normalizing_constants}. It results in the following constants.
\begin{eqnarray}\label{eqn: a_n normalized}
   a_{\{K,r\}} &=& \frac{2}{K} \Gamma(r) \exp\left\{Q^{-1}\left(r,\frac{1}{K}\right)\right\} Q^{-1}\left(r,\frac{1}{K}\right)^{1-r}
\\
\label{eqn: b_n normalized}
  b_{\{K,r\}} &=& 2Q^{-1}\left(r,\frac{1}{K}\right) + o(a_{\{K,r\}}),
\end{eqnarray}
where $Q^{-1}\left(r,\frac{1}{K}\right)$ is the inverse of the regularized upper incomplete gamma function, that is, $Q(r,z)= \frac{\Gamma(r,z)}{\Gamma(r)}$, and the inverse is defined with respect to $z$. Accordingly, throughout this paper, whenever the normalizing constants are required (e.g., all simulation results), we utilize the constants above.
It is important to note that the main results of this paper, which are the scaling laws, are not affected by the constants used, as long as they have the same asymptotic behaviour. The difference is only in the speed of convergence.

\label{comment_on_fig_1}Figure \ref{fig:max_norm_dist30} evaluates the expression for the distribution of the maxima of $K$ independent variables (with a $\chi^2$-distribution each), using the normalizing constants in \eqref{eqn: a_n normalized} and \eqref{eqn: b_n normalized}, and taking $o(a_{\{K,r\}})$ to be $0$. In particular, we start from the convergence result in \eqref{eq: Gumbel distribution}. This result holds, of course, when $K\to\infty$. However, letting $\zeta = a_K z + b_K$, and using this change of variables, the distribution for finite $K$ can be expressed as $\Pr\left(\bM_K \leq \zeta \right) = \e^{-\e^{-(\zeta- b_K)/a_K}}.$ Of course, differentiating this w.r.t.\ $\zeta$ provides a probability density function (PDF) for any finite $K$. This distribution is depicted in the figure (solid line) for $30$ chi-squared random values (i.e., $K=30$) with $r=4$ (left plot) and $r=16$ (right plot) degrees of freedom. We compare these analytical results with simulation results (presented by bars) which are generated by randomly generating $K=30$ random values, drawn from the $\chi^2$-distribution and selecting the maximal value for each such instance, repeating the process for $10,000$ instances, and presenting the histogram of the outcome (i.e., the empirical PDF). As can be seen in the figure, the agreement between the simulation results and analytical results utilizing the normalizing constants suggested in \eqref{eqn: a_n normalized} and \eqref{eqn: b_n normalized} is evident even for a small number of random variables. For comparison, we also depict the analytical results obtained by utilizing the normalizing constants previously suggested in the literature. In particular, the dashed lines depict the resulting Gumbel distribution under the normalizing constants given in \eqref{eqn: a_n normalized slow} and \eqref{eqn: b_n normalized slow}. As can be seen in the figure, the maximal value distribution is not well approximated by the Gumbel distribution and these normalizing constants for moderate $K$ ($K=30$). Obviously, for asymptotically large $K$ both sets of normalizing constants provide good approximation, converging to the Gumbel distribution at the limit (a proof is given in \Cref{sec:proofs}).   
\begin{figure}
	\centering
		\includegraphics[scale=0.6]{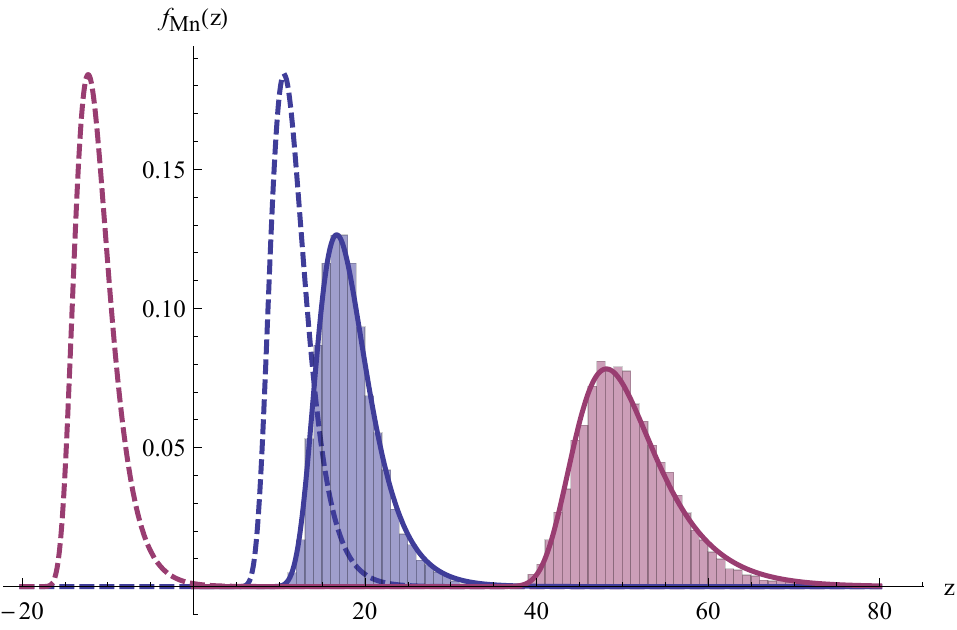}
	\caption{Convergence of the distribution of the maximum to the analytical result. Bars depict the maximal channel norm values among $K=30$ channels, according to our simulation results, while the solid lines depict the max norm Gumbel distribution with $4$ and $16$ degrees of freedom, under the normalizing constants derived in \eqref{eqn: a_n normalized} and \eqref{eqn: b_n normalized}. The dashed lines depict the max norm Gumbel distribution with the same degrees of freedom, but under the normalizing constants that are found in the literature, and are given in \eqref{eqn: a_n normalized slow} and \eqref{eqn: b_n normalized slow}.  }
	\label{fig:max_norm_dist30}
\end{figure}
} 

\subsection{Linear Receivers}\label{subsection linear receivers}
\blue{
In this paper we aim to analyse the ergodic sum-rate of a practical system in which only a subset of the users is scheduled at each time slot and only linear decoders are considered at the BS. In particular,  we focus on either the ZF receiver (Section~\ref{sec. dist alg}) or the MMSE receiver (Section~\ref{sec. MMSE rec}). We further assume optimal coding of the resulting single user Gaussian channels, given the effective Signal to Noise Ratio (SNR).

Specifically, for the ZF receiver, focusing on the signal received from the $i$-th user, rewrite \eqref{eq: channel model} as:
\begin{equation*}
\by = \bh_i \bx_i + \sum_{\ell \neq i}^{j} \bh_\ell \bx_\ell + \bw.
\label{eq: channel model2}
\end{equation*}
Let $\bV_i$ be a unitary matrix representing the null space of the subspace spanned by $\{\bh_\ell\}_{\ell \neq i}$. Since the entries of the channel vectors are i.i.d., when $j$ users transmit the subspace spanned by the vectors $\{\bh_\ell\}_{\ell \neq i}$ has rank $j-1$ with probability one \cite[Ch. 8.3.1]{tse2005fundamentals}. Thus, to null the inter-stream interference, the receiver projects the received vector $\by$ on the subspace spanned by $\bV_i$, and detects stream $i$ by match filtering the resulting vector $\bV_i \bh_i$.  Consequently, the ZF receiver is the vector $\bV_i^\dagger \bV_i \bh_i$, which is the closest direction within the subspace $\bV_i$ to $\bh_i$. Note that this channel inversion technique is simply the $i$-th row vector of the pseudo-inverse matrix $(\bH^\dagger \bH)^{-1}\bH^\dagger$. The effective channel gain in this case is $\Vert \bV_i \bh_i \Vert^2$. Note also that a full degrees-of-freedom gain is attained when  $r$ users transmit. In this case, $\dim(\bV_{i})=1$,  and $\Arrowvert \bV_i \bh_i \Arrowvert^2 \sim \chi^2_{2}$, e.g., \cite{tse2005fundamentals}. Accordingly, when using a ZF receiver, we aim at algorithms which select at most $r$ users (of the available $K$) at each time slot. The above discussion also stresses out the point that the number of transmitting users, $j$, influences the distribution of $\bV_i$. This will be crucial in our future derivations. Thus, hereafter, we will use the notation $\bV_i(j)$ to emphasise this fact.

\label{commentr23}In general, a scheduler with CSI is defined by a function $\mathcal{I} = \mathcal{I}(\bH)$ that maps a set of $K$ channels $\bH$ to an index subset $\mathcal{I} \subset [1:K]$ of cardinality less than or equal to $r$, which reflects the set of selected users. Assume $|\mathcal{I}(\bH)|=j$. The ergodic rate of user $i$ under such scheduling is
\begin{equation}
\mathcal{R}_i = \E\brkt{\log\prnt{1+P \Vert \bV_i(j) \bh_i \Vert^2 }\times \indicator{i \in \mathcal{I}(\bH)}},
\label{eqn:ergodic_rate}
\end{equation}
where $\indicator{i \in \mathcal{I}(\bH)}$ is the indicator function of the event that user $i$ is selected with channels state $\bH$, along with other users in the subset $\mathcal{I}(\bH)$. In the ZF detection technique,  $\bV_i(j)$ captures the influence of the other scheduled users on user $i$'s instantaneous rate, according to the correlation between the channels of the users in $\mathcal{I}(\bH)$.  
 Consequently, the ergodic rate of user $i$, as defined in \eqref{eqn:ergodic_rate}, not only depends on whether user $i$ is scheduled for transmission or not, but also on the channel state of all other scheduled users. Both aspects are determined based on the scheduler policy. For example, in a round robin scheduler, user $i$ is selected whenever its turn has arrived. As far as user $i$ is concerned, the other users are selected arbitrarily, independent from its channel. On the other hand, a different scheduler may select users based on some channel criteria (e.g., channels correlation, channels gain, fairness, etc.), which results in a different (ergodic) rate. 

Summing over all users and taking the expectation with respect to the channel state matrix ($\bH$) distribution, the ergodic sum rate is
\begin{align*}
\mathcal{R} &= \sum_{i=1}^K \E \brkt{\log\prnt{1+ P \Vert \bV_i(j)  \bh_i \Vert^2} \times \indicator{i \in \mathcal{I}(\bH)}}&\\
&= \E \brkt{\sum_{i=1}^K  \log\prnt{1+ P \Vert \bV_i(j) \bh_i \Vert^2} \times \indicator{i \in \mathcal{I}(\bH)}}&\\
&= \E  \brkt{\sum_{i \in \mathcal{I}(\bH)} \log\prnt{1+ P \Vert \bV_i(j)  \bh_i \Vert^2}}.
\end{align*}
\label{commentr19}Accordingly, considering all index subsets $\mathcal{I}$ with size $\vert \mathcal{I} \vert =j \leq r$, the optimal scheduler maximizes the above expectation. That is, the maximization is over all \emph{selection functions} $\mathcal{I}(\bH)$, which yield subsets of cardinality $\leq r$. Accordingly, the (ergodic) sum rate of the optimal scheduler in this case is
\[
\mathcal{R} = \max_{\mathcal{I}(\bH)} \E  \brkt{\sum_{i \in \mathcal{I}(\bH)} \log\prnt{1+ P \Vert \bV_i(j)  \bh_i \Vert^2}}.
\]
Obviously, this maximization is hard to solve as it needs to examine all possible sets of users which are smaller than or equal $r$. Furthermore, such scheduler relies on knowledge of the matrix $\bH$, i.e., acquiring the channel state from all users. Accordingly, substitute schedulers are considered in practice, which try to minimize the overhead and complexity of such a centralized process, and select a group of users, preferably in a distributed manner, approximating the optimal selection.
}

While simple and intuitive, the ZF receiver is limited in its performance. The MMSE receiver, however, although still linear, maximizes the mutual information and hence achieves better performance  (e.g. \cite{tse2005fundamentals,li2006distribution,kim2008performance}). In this receiver, to  detect data stream $i$, the receiver treats the rest of the  streams as noise. It then whitens the resulting colored noise and uses a matched filter to obtain maximum SINR.

Let $\bH$ be a  matrix whose columns are the channel vectors of the transmitting users. Similarly, let $\bH_{(-i)}$ be the matrix  $\bH$ with its $i_{th}$ column removed and define
\begin{equation}\label{R in mmse}
\bR = \left(\blue{P} \bH_{(-i)}\bH_{(-i)}^{\dagger}	 + I\right)^{-1}.
\end{equation}
Then, the corresponding output $\text{SINR}_i$ on the stream $i$ can be expressed by \cite{kim2008performance}:
\begin{equation}\label{sinr in mmse}
\text{SINR}_i = \blue{P} \bh_i^{\dagger}\bR \bh_i.
\end{equation}
This SINR value will be at the basis of our analysis in \Cref{sec. MMSE rec}.

\section{A Distributed Algorithm}\label{sec. dist alg}
\blue{The performance of the system, which only schedules a subset of users in each transmission opportunity, is highly dependent on the user selection procedure. In particular, two main aspects will highly influence the expected sum rate of such a system: (i) the quality of the channel between each selected user and the BS (ii) the inter-user interference between the selected users. The challenge is, hence, how to schedule users opportunistically, without the burden of pooling CSIs from many users or a tedious negotiation process prior to each transmission. A common distributed approach for single user selection, which aims to select only a single user opportunistically, is a threshold-based procedure, in which a capacity threshold is set, and only a user whose channel capacity exceeds it transmits (\cite{kampeas2012opportunistic,qin2003exploiting}). The algorithm examined herein, adopts a similar threshold-based approach to select a \emph{group of users}.

Determining the threshold raises several important questions. For example, on which variable should a threshold be set? What should be the threshold value and, specifically, how many users are expected to pass it on average? How will a user which exceeded the threshold determine its transmission rate, etc. Recall that the user attainable rate is influenced not only by its own channel but also by the mutual interference from the other transmitting users. In the sequel, we address the above questions. 

The Threshold-Based (TB) Channel Access algorithm we analyze, denoted {\scshape TB-Channel-Access}, works as follows: given the number of users $K$, we set a threshold $u_k$ on the \emph{channel norm} (gain) $\Vert \bh \Vert^2$. We assume that prior to each transmission opportunity, the BS sends a pilot signal from each of its $r$ antennas. A user with a channel norm greater than the threshold transmits. We further assume that the BS can estimate accurately the channels $\bh_i$ of the above-threshold users directly from their transmitted signals, e.g., by utilizing the signaling schemes in \cite{xu2015analysis,suh2003preamble,magistretti2014802}. 

Yet, a linear receiver cannot recover more than $r$ data streams simultaneously. That is, the event that more than $r$ users begin transmission simultaneously (which we term ``collision") results in an unsuccessful transmission attempt, and zero information is extracted from the received signal. Thus, the threshold should be set such that no more than $r$ users will exceed it on each transmission opportunity. On the other hand, if no user attempts transmission in a slot, it will also be wasted (remain idle). Accordingly, throughout this paper we say that a slot is utilized if at least one user, but no more than $r$ users, are transmitting, and unutilized otherwise. Note that the target number of users that their channel gain exceed the threshold on average (denoted by $k$), which sequentially determines the threshold itself, should be optimized in order to minimize the number of unutilized slot. 

In case that the users are not i.i.d., and specifically if the users' channel distributions are not identical, one might set different thresholds, such that the access probability remains fair. On the other hand, one might use a single threshold, and achieve better performance at the price of possible starvation of the weak users. \label{comment_r13}\blue{Clearly, the EVT results for i.i.d.\ random variables do not apply directly to the non i.i.d.\ case. However, in \cite{kampeas2012capdisthetnet}, a more refined version, using point process approximation for non-identically distributed variables was used to analyse the throughput in such a case. While the model therein assumed single-user scheduling, without intricate interference and multi-user issues, and hence is simpler than the one we consider here, some of the methods used can be applied to this problem, and shed light on the non i.i.d.\ case.}

The next challenge is in determining the rate that an above-threshold user (a self-scheduled user) should transmit. Note that for user $i$ knowing its own channel vector $\bh_i$, is not sufficient to determine its transmission rate, as it must also know the actual number of exceeding users, $j$, and their channel vectors, in order to achieve the rate $\log\prnt{1+P\Vert \bV_i(j) \bh_i \Vert^2}$ at a specific slot. Specifically, the rate that each above-threshold user $i$ can achieve under ZF depends on both its channel gain, $\Vert\bh_i \Vert^2$, \emph{and its correlation with the channels of other above-threshold users}, via the matrix $\bV_i(j)$. In this section, we analyze the {\scshape TB-Channel-Access} performance under two transmission schemes, which address the above challenge. In the first option, users transmit at their ergodic rate \cite{caire1999capacity}. \label{comment E8}That is, each user codes over a sequence of slots (blocks) \emph{in which it has above-threshold channel gain}. The BS waits until it collects the required number of blocks, and decodes based on all blocks of all above-threshold users. While this scheme can allow each user to achieve the ergodic rate $\E\brkt{\log\prnt{1+P\Vert \bV_i \bh_i \Vert^2}\big| \Vert\bh_i\Vert^2 > u_k}$ at the slots in which the user transmitted, without exchanging data about the rate each user should use, it may suffer from a large decoding delay. \label{commentr11part2}Thus, to evaluate the benefit of rate coordination, we provide a second, supplementary scheme, in which prior to each transmission, the BS sends all above-threshold users feedback (via, e.g., a broadcast message) that indicates their rate. Consequently, in this case, the user \emph{selection process} is distributed (i.e., the users still decide whether to transmit or not distributively, without exchanging any information between themselves or with the BS), yet the transmission rates are coordinated between the BS and the \emph{selected users}. Note that under this approach the BS can also avoid the collision slots by singling only a subset of the threshold exceeding users to transmit, in case more than $r$ exceeded, avoiding the high price of a collision slot.}

The first result, \Cref{C ZF uniform users exact} below, gives the \BLUE{ergodic rate seen by a user, as well as the} sum-rate under the above distributed user selection algorithm and ZF decoding.
Remember that $\bV_i(j)$ is determined by the channels $\{\bh_\ell\}_{\ell\neq i}$. \blue{I.e., it depends on the channels of the other exceeding users, and in particular, on the number of exceeding users, $j$.} Consequently, in the absence of knowledge on the other above-threshold channels, user $i$ cannot resolve its instantaneous rate, and as a result, it must code over long sequences of slots. 
 Note that this simple proposition still includes an expectation on the channel vectors seen by the users, hence cannot give the understanding we wish regarding the \BLUE{ergodic rate and} sum-rate under the suggested algorithm. Still, it will be the starting point, from which we will derive the bounds which give the right insight and scaling laws.

\begin{proposition}\label{C ZF uniform users exact}
For $K \gg r$, the \RED{expected achievable} \BLUE{ergodic} sum-rate of Algorithm {\scshape TB-Channel-Access} with ZF detection is given by
\begin{align*}
\mathcal{R}(u_k) &= \sum_{j=1}^{r}\frac{k^{j}e^{-k}}{j!} \sum_{i=1}^{j} \E\left[\left.\log\Big(1 + P \Arrowvert \bV_i\blue{(j)}\bh_i \Arrowvert^2 \Big) \right|  \Arrowvert \bh_i \Arrowvert^2 > u_k \right]+ O\left(\frac{\log \log K}{K}\right),
\end{align*}
where $k$, to be optimized, is the expected number of users to exceed the threshold $u_k$, the $\{\bh_i\}_{i=1}^{j \leq r}$ are the channel vectors of the users who exceeded the threshold and $\{\bV_i\}_{i=1}^{j\leq r}$ are the corresponding null spaces.
\end{proposition}
\BLUE{Note that in analogy to the upper bound in \cite{jagannathan2007scheduling} for the downlink scenario, the above sum rate can be easily upper bounded by $r \E \log\prnt{1+P\Vert \bh_{(1)}\Vert^2}$, where $\bh_{(1)}$ is the channel vector with the largest norm. Specifically, this bound is derived by assuming exactly $r$ users transmit, neglecting any inter-user interference and taking all norms to be equal to the largest one. Obviously, such a bound is very loose.} To ease notation, the $O\left(\frac{\log \log K}{K}\right)$ approximation error is omitted from now on. Clearly, this gap is negligible compared to the leading terms in the expression.
\begin{proof}
According to the law of total probability, we express the \RED{expected achievable} \BLUE{ergodic} sum-rate in a slot by summing over the number of users who exceed the threshold, and the sum-rate these users see, given that they exceeded the threshold. As mentioned, if more than $r$ users are transmitting in a slot, the \red{receiver}\blue{BS} cannot successfully null the inter-stream interference, and the sum-rate in that slot is assumed to be zero.

Hence, the expected sum-rate has the form:
$$\sum_{j=1}^{r}\Pr\{ j \textmd{ users exceed}\}\sum_{i=1}^{j} \text{E}\left[\left. C_i \right| \Arrowvert \bh_i \Arrowvert^2 > u_k \right].$$
When the users are i.i.d., the probability of $j$ threshold exceedances follows the binomial distribution with probability $p= k/K$ to exceed the threshold. Since we consider large $K$ and small values of $k$, the number of users to exceed the threshold can be approximated by the Poisson distribution with an approximation error in the order of $1/K$.
 However, as this approximation error is within the sum, it is multiplied by the individual capacities, which scales like the optimal scaling law of the multi-user diversity when a single, strongest user, is scheduled, i.e., $\Theta\left(\log \log K \right)$ (see e.g., \cite{hassibi2007fundamental,yoo2006optimality} and references therein). Hence, the  $O\left(\frac{\log \log K}{K}\right)$ approximation error.
Finally, note that the number of exceeding users $j$ affects the effective SNR seen by the attending users. In particular, when $j$ users exceed threshold, the dimension of $\bV_i$ is $(r-j+1) \times r$ \cite[Chapter 8.3.1]{tse2005fundamentals}. Thus, as $j$ decreases, the signal of the attending users is projected on a less restrictive null-space. Accordingly, each stream may spread on more receiving antennas in the ZF process, which leads to a higher power gain\RED{(for details on ZF decoding, see \mbox{\cite{tse2005fundamentals}})}. Nonetheless, the reader should not be confused. The highest \RED{capacity} \BLUE{sum-rate} is attained when $r$ users utilize the channel simultaneously to achieve a full degrees-of-freedom gain.
\end{proof}

To evaluate the result in \Cref{C ZF uniform users exact}, the behavior of $\Arrowvert \bV_i\blue{(j)} \bh_i \Arrowvert^2$ should be understood, especially considering the fact that the number of users exceeding the threshold is random. To this end, the following upper and lower bounds are useful. These bounds will be the basis of the scaling laws we derive.
\begin{lemma}\label{C ZF uniform users upper}
The \RED{expected achievable} \BLUE{ergodic} sum-rate of Algorithm {\scshape TB-Channel-Access} with ZF detection satisfies the following upper bound
\begin{equation*}
\mathcal{R}(u_k) \leq \sum_{j=1}^{r}\frac{k^{j}e^{-k}}{j!} j \log\Big(1 + \frac{P}{r}(r-j+1)(u_k + a_{\{K,r\}} )\Big),
\end{equation*}
where $a_{\{K,r\}}$ is given by \eqref{eqn: a_n normalized} and $u_k$ is the threshold\blue{, set such that the number of users that exceed it on average is $k$}.
\end{lemma}
The bound in \Cref{C ZF uniform users upper}, while not giving the exact sum-rate, still depicts the essence of the system behavior.
To understand its implications, we note the following: We set a threshold such that $k$ out of the $K$ users exceed it on average. I.e., the average exceedance rate is $k$. Indeed, the expression $\frac{k^{j}e^{-k}}{j!}$ in the sum over $j$ gives the probability for exactly $j$ users exceeding.
Each of the $j$ users, under zero forcing, experiences a single user channel, with its power $P$ scaled according to two factors: (i) a multiplication by $(u_k+a_{\{K,r\}})$, as this is the average norm of its channel vector, where  $u_k$ is the threshold exceeded, and $a_{\{K,r\}}$ is the average distance \emph{above} the threshold. (ii) a multiplication by $\frac{r-j+1}{r}$, as in case only $j<r$ users exceeded the threshold, the zero forcing algorithm does not have to cancel $r-1$ users, only $j-1$, hence the null space has a larger dimension, yet the number of receive antennas is $r$. As the threshold $u_k$ will be shown to be $\Theta(\log K)$, the optimal scaling law will follow. A complete discussion will be given after the lower bound is introduced.
Indeed, as it turns out in the simulation results, the bound in \Cref{C ZF uniform users upper} is tight even for relatively small number of antennas and users.
\RED{Note, however, that if less than $r$ users exceed, as a higher SNR can be attained at the receiver, in order to achieve capacity in this case, a user must know how many users exceeded, so it can exploit the high SNR for, e.g., higher transmission rate. Hence, we require that the number of users that actually exceeded the threshold will be announced. \mbox{\SEFI{Can't we take it in ergodic sense as well ?}}}

\begin{remark}
Note that in practice it is beneficial to choose $k$ slightly smaller than the number of antennas $r$. This is since if less than $r$ users exceed, the SNR seen by each user is only larger, yet if more than $r$ users exceed, the slot is lost.
\end{remark}
\begin{proof}[Proof (\Cref{C ZF uniform users upper})]
We start with \Cref{C ZF uniform users exact}.
By Jensen inequality,
\begin{eqnarray}
\mathcal{R}(u_k) &=& \sum_{j=1}^{r}\frac{k^{j}e^{-k}}{j!} \sum_{i=1}^{j} \text{E}\left[\left.\log\Big(1 + P \Arrowvert \bV_i\blue{(j)} \bh_i \Arrowvert^2 \Big)\right|  \Vert \bh_i \Vert^2 > u_k\right] \blue{+ O\prnt{\frac{\log \log K}{K}}} \nonumber\\
&\leq& \sum_{j=1}^{r}\frac{k^{j}e^{-k}}{j!} \sum_{i=1}^{j} \log\Big(1 + P \text{E}\left[\left. \Arrowvert \bV_i \blue{(j)} \bh_i \Arrowvert^2 \right| \Vert  \bh_i \Vert^2 > u_k \right] \Big) \blue{+ O\prnt{\frac{\log \log K}{K}}}.\label{eq. in upper bound proof}
\end{eqnarray}
Consider the norm $\Arrowvert \bV_i \blue{(j)} \bh_i \Arrowvert^2$, where $\bV_i\blue{(j)}$ has $r-j+1$ rows. Denote by $\bV_i^{(m)}$ the $m$th row of \blue{$\bV_i(j)$ (ommiting the dependence on $j$ for clarity)}, and let $\vert \langle \cdot , \cdot \rangle \vert^2$ denote the squared inner product. We have
\begin{eqnarray*}
\text{E} [\left. \Vert \bV_i\blue{(j)} \bh_i \Vert^2 \right|  \Vert \bh_i \Vert^2 > u_k] &=& \text{E}\left[\left.\sum_{m=1}^{r-j+1} \vert \langle\bV_i^{(m)},\bh_i\rangle \vert^2 \right|  \Vert \bh_i \Vert^2 > u_k\right]\\
&\stackrel{(a)}{=}& \sum_{m=1}^{r-j+1} \text{E}\left[\left. \Arrowvert \bh_i \Arrowvert^2 \frac{\vert \langle\bV_i^{(m)},\bh_i\rangle \vert^2}{\Arrowvert \bh_i \Arrowvert^2 \Arrowvert \bV^{(m)}_i \Arrowvert^2} \right| \Vert \bh_i \Vert^2 > u_k\right]\\
&\stackrel{(b)}{=}& \text{E}[\left.\Arrowvert \bh_i \Arrowvert^2 \right| \Vert \bh_i \Vert^2 > u_k] \sum_{m=1}^{r-j+1} \text{E} \left[ \frac{\vert \langle\bV_i^{(m)},\bh_i\rangle \vert^2}{\Arrowvert \bh_i \Arrowvert^2 \Arrowvert \bV^{(m)}_i \Arrowvert^2}\right]\\
&\stackrel{(c)}{=}&  \text{E}[\left.\Arrowvert \bh_i \Arrowvert^2 \right| \Vert \bh_i \Vert^2 > u_k](r-j+1) \int_{0}^{1}(1-\alpha)^{r-1}d \alpha\\
&\stackrel{(d)}{=}&  (u_k + a_{\{K,r\}} )(r-j+1) \frac{1}{r}.
\end{eqnarray*}
In the above chain of equalities, (a) is since $\Arrowvert \bV^{(m)}_i \Arrowvert^2=1$ (b) is since $\bh_i$ is a random i.i.d.\ complex normal vector, and the squared-normalized inner product $\frac{\vert \langle\bV_i^{(m)},\bh_i\rangle \vert^2}{\Arrowvert \bh_i \Arrowvert^2 \Arrowvert \bV^{(m)}_i \Arrowvert^2}$ is its angle from $\bV_i^{(m)}$, a vector in the null space of $\{\bh_l\}_{l\ne i}$. Since these vectors are independent of $\bh_i$, this angle is independent of the norm of $\bh_i$ (c) is since the distributions of the norms and angles are independent of $m$, and since, by \cite[Lemma 3.2]{jagannathan2006efficient}, the angle has the same distribution as the minimum of $r-1$ independent uniform $[0,1]$ random variables (i.e., with CDF $1-(1-\alpha)^{r-1}$, $0 \leq \alpha \leq 1$) (d) is the result of computing the expected norm of an i.i.d.\ complex normal random vector, \emph{given that the norm is above a threshold $u_k$.} The details are in \Cref{coro: expected capacity above est by GEV thr}, \Cref{sec: tail dist}.

Substituting in \eqref{eq. in upper bound proof}, \blue{and omitting the $O\prnt{\frac{\log \log K}{K}}$, we obtain the result of Lemma~\ref{C ZF uniform users upper}.} \redo{we have
\begin{multline*}
\mathcal{R}(u_k) \leq \sum_{j=1}^{r}\frac{k^{j}e^{-k}}{j!} \sum_{i=1}^{j} \log\Big(1 + \frac{P}{r}(u_k + a_{\{K,r\}})(r-j+1) \Big)
\\
\leq \sum_{j=1}^{r}\frac{k^{j}e^{-k}}{j!} j \log\Big(1 + \frac{P}{r}(u_k + a_K)(r-j+1)\Big),
\end{multline*}
which completes the proof.}
\end{proof}
We now present a corresponding lower bound.
\begin{lemma}\label{C ZF uniform users lower}
The \RED{expected achievable} \BLUE{ergodic} sum-rate of Algorithm {\scshape TB-Channel-Access} with ZF detection satisfies the following lower bound.
\begin{equation*}
\mathcal{R}(u_k) \ge \left( \sum_{j=1}^{r}\frac{k^{j}e^{-k}}{j!} j\right)(r-1) \int_0^1 (1-\alpha)^{r-2} \log\left(1 + P u_k \alpha\right) d\alpha,
\end{equation*}
where $u_k$ is a threshold\blue{, set such that the number of users that exceed it on average is $k$}.
\end{lemma}
It is important to note that the integral in \Cref{C ZF uniform users lower} above has a finite series expansion with $r$ summands. This finite series has $\log(1+Pu_k)$ at the leading term, resulting in the expected scaling law. We describe it in \Cref{claim. integral result} below, within the proof of the main result in this section - \Cref{theorem: C ZF uniform users - scaling}. The proof of \Cref{C ZF uniform users lower} is deferred to \Cref{proof:ZF_lower}.

The results above lead to the following scaling law, which is the main result in this section. It asserts that the scaling law of $r \log (P\log K)$ for the sum rate in a multi-user system can in fact \emph{be achieved distributively}, without collecting all channel states from all users and scheduling them in a centralized manner. In other words, the {\scshape TB-Channel-Access} algorithm suggested selects an optimal \emph{set of users} (asymptotically in the number of users) distributively and without any cooperation. This is summarized in the next theorem.
\begin{theorem}\label{theorem: C ZF uniform users - scaling}
The \RED{expected achievable} \BLUE{ergodic} sum-rate of Algorithm {\scshape TB-Channel-Access} with ZF detection scales as $\Theta(r \log(P \log K))$ for large enough number of users $K$.
\end{theorem}
\begin{proof}
By \Cref{C ZF uniform users upper},
\begin{eqnarray*}
\mathcal{R}(u_k) &\leq& \sum_{j=1}^{r}\frac{k^{j}e^{-k}}{j!} j \log\Big(1 + P\frac{(r-j+1)(u_k + a_K)}{r}\Big)
\\
&\leq& \sum_{j=1}^{r}\frac{k^{j}e^{-k}}{j!} j \log\Big(1 + P(u_k + a_K)\Big)
\\
&\leq& r \log\Big(1 + P(u_k + a_K)\Big).
\end{eqnarray*}
\label{comment E10}On the other hand, consider the lower bound given in \Cref{C ZF uniform users lower}. \blue{We wish to show that the sum in the parenthesis is bounded from below by a constant times $r$. First, since $k$ is a parameter to be optimized, the optimum is at least as large as when choosing $k=r$. The resulting term is therefore $\sum_{j=1}^{r}\frac{r^{j}e^{-r}}{j!} j$. This sum is simply $r$ times \emph{a sum over the Poisson distribution (with parameter $r$), up to $r-1$.} In \Cref{sec. poisson monoticity} we provide a proof that this sum is monotonically increasing in $r$, and therefore can be bounded from below by taking $r=1$. We thus have,}
\begin{equation*}
\sum_{j=1}^{r}\frac{r^{j}e^{-r}}{j!} j = r\sum_{j=0}^{r-1}\frac{r^{j}e^{-r}}{j!} \mbox{\red{j}} > \blue{r \e^{-1}}.
\end{equation*}
\red{where the last inequality is by evaluating the sum at $r=2$. }Note that \blue{this gives a bound of $0.37 r$. Evaluating at $r=2$ gives $0.4 r$, and} larger values of $r$ give only slightly larger values, with a limit of $0.5$ as $r\to\infty$.\footnote{This is the CDF of a Poisson random variable with parameter $r$, calculated at $r-1$. The limiting behavior can be found in \cite{janssen2008gaussian}.}

Now, consider the integral over $\alpha$ in \Cref{C ZF uniform users lower}. In \Cref{proof:integral_result}, we prove the following claim.
\begin{lemma}\label{claim. integral result}
The integral over $\alpha$ in \Cref{C ZF uniform users lower} has the following finite series expansion:
\begin{multline*}
(r-1) \int_0^1 (1-\alpha)^{r-2} \log\left(1 + P u \alpha\right) d\alpha\\
 = \left(\frac{1 + P u}{P u}\right)^{r-1}\log\left(1+P u\right)-\sum_{i=0}^{r-2} \left(\frac{1+ P u}{P u}\right)^i \frac{1}{r-1-i}.
\end{multline*}
\end{lemma}
\off{ 
This integral is a know hypergeometric function (which can be handled in Mathematica).
\begin{equation*}
(r-1) \int_0^1 (1-\alpha)^{r-2} \log\left(1 + P u \alpha\right) d\alpha
= \frac{u r!}{r}\sum_{i=0}^{\infty}\frac{i!(-u)^i}{(i+r)!}.
\end{equation*}
The resulting expression can be symbolically evaluated for integer values of $r$ (the number of antennas).}

This gives a finite series expansion for the integral, in terms of the power $P$ and the threshold $u$. For example, for $r=4$ we have
\begin{multline}
\label{for r=4}
(r-1) \int_0^1 (1-\alpha)^{r-2} \log\left(1 + Pu\alpha\right) d\alpha \\ =
\frac{6(1+Pu)^3\log(1+Pu)-2u^3P^3-3u^2P^2(1+Pu)-6uP(1+Pu)^2}{6 u^3 p^3}.
\end{multline}
Thus, the integral can be easily approximated by $\log \left(1 + P u\right) + O\left(\frac{\log u}{u}\right)$.

Since we consider the regime of large enough number of users $K$, yet a finite number of antennas $r$, we have $k = O(1)$ and as a result $u_k = \Theta(\log K)$. In fact, since the distribution of the projected channel gain seen by a user in this detection scheme is the exponential distribution with rate $1/2$, it can be shown that $u_k = 2(\log K - \log k)$ (we discuss the threshold value in detail in \Cref{sec: thr est}). This gives rise to the $\Theta(r \log (P\log K))$ scaling law.
\end{proof}

\begin{remark}
When a long-term power control is considered, an additional power gain factor of $K$ can be attained. In particular, taking into account the $r/K$ probability of each user to exceed the threshold and utilize the channel in a slot, a user can transmit at instantaneous power $P' \sim (K/r) P$, and still satisfy the long-term power constraint. Thus, when $K\gg r$, the resulting \RED{capacity} \BLUE{rate} scales as $r \log (K \log K)$, which is \emph{much better}. In other words, since a user keeps silence for roughly $K/r$ slots on average, until it exceed the threshold, it can use very high instantaneous power.
Furthermore, the wide dynamic rage of a wireless cellular device, which is about $90$dB, allows a regime with a wide range of values for $K$ and relatively moderate values of $r$, and thus, the additional power gain factor $K$ may be achievable in this regime.
\end{remark}

\subsection{Distributed User Selection with Feedback}\label{dus_feedback}
\blue{As previously mentioned, in order to achieve the ergodic rate, a sender needs to code over many signal blocks, each of which having an above-threshold channel vector and no more than additional $r-1$ such above-threshold users. When the number of users $K$ is large and the threshold exceedance probability is small, such a procedure may be prohibitively long. Specifically, when the probability of a user to exceed the threshold is $k/K$, the average interval between consecutive exceedances of a user is $K/k$, and the average number of slots that are required for transmitting $N$ blocks is $N K/k$, assuming there are no unutilized slots. Obviously, taking into account unutilized slots (collisions and idle slots), this number can grow depending on the probability of such event; increasing the threshold exceedance probability (increasing $k$) reduces the time interval between the user's consecutive transmission attempts, yet increases the number of collisions and vice versa, decreasing the threshold exceedance probability (decreasing $k$) enlarges the time interval between the user's consecutive transmission attempts, yet decreases collision probability and increases idle slot probability.}

\label{comment r11}\blue{Besides the aforementioned latency issue, not knowing the exact transmission rate on a per slot basis requires the user to transmit in its expected ergodic rate, and the BS to completely drop slots in which too many users transmitted. In this subsection, we investigate the benefits of rate coordination, which mitigate the above issues. In particular, we examine the attainable expected rate when each transmitting user can set its rate based on the transmitting set (the above-threshold users) in each of its transmission attempts. Obviously, such mechanism requires coordination between the users, which can be attained by a \emph{feedback from the BS to the transmitting users}, prior to each transmission, notifying each transmitting user its transmission rate based on the other transmitting users. Even though this subsection aims at evaluating the gain\slash loss due to such coordination or lack thereof, it is important to note that as far as practicality is concerned, such a mechanism is relatively easy to attain. For example, a procedure in which the above-threshold users signal the BS their intention to transmit, the BS estimates the channel vector between each such user and itself based on these transmitted signals and based on this channel state information notifies the selected users (e.g., via a broadcast message) their transmission rate. Obviously, the feedback process must be \emph{sufficiently short compared to the coherence time}. In other words, the slot duration which includes, besides the data transmission itself, a short preceding process, in which the above-threshold users estimate and signal the BS their channel gain, and the BS replies with the rate each should transmit, is within the fading block duration. Note that this rate-coordination process \emph{signals only the rates} back to the users and not complete CSI, hence requires only a small number of feedback bits per block (see e.g., \cite{lau2004design}). Further, note that a similar process is implemented by various standards and protocols such as the IEEE 802.11. Moreover, this signaling process is done only with self-selected users. Thus, the user selection process is still distributed in its essence, with only a limited amount of feedback sent by the BS. Obviously, such feedback not only enables us to announce the instantaneous rates users should use, but also reduces the decoding delay, as each user can code over a single fading block. Since as previously mentioned, in the context of this paper the main motivation of analyzing this process is to quantify the rate-coordination gain, we omit any further discussion regarding the exact procedure utilized to attain such feedback, and concentrate on the analysis of the attainable rate.} 

\off{In the following, we examine the performance when feedback from the \red{receiver}\blue{BS} to the transmitting users is allowed \blue{before transmission}. Obviously, such a feedback enables us to announce the instantaneous rates \blue{the users should use}, and hence, reduce the decoding delay, as each user can code over a single fading block\red{, without affecting the rates at all}.

\blue{Specifically, since channels remain quasi-static within fading blocks, yet drawn independently between blocks, we assume the feedback process is sufficiently short compared to the block duration. In other words, transmission of the selected users is preceded by a short process, in which users estimate their channel gains, above the threshold users signal the BS of their gains, and the BS replies with the rates at which these users should transmit. In the limit of a large block length, this rate-coordination process is assumed to be finite, with a small number of feedback bits per block (see e.g., \cite{lau2004design}). Note that the process coordinates only rates back to the users, and not complete CSI. Moreover, this is done only to the selected users. Thus, the main user selection process is still distributed in its essence, with only a limited amount sent by the BS. Still, analysing such a process is important, as it quantifies the rate-coordination gain.}
}

Thus, a feedback that is sent before each transmission to the above-threshold users can indicate to each above-threshold user its transmission rate in the upcoming slot. 
\blue{Moreover, the feedback can even increase the expected sum-rate, as it allows the BS to avoid unutilized slots due to collision (i.e., when more than $r$ users have exceeded the threshold), by notifying some of the above-threshold users not to transmit}.
In particular, the gain in performance is the following.
\begin{corollary}\label{coro: delta sum-rate with feedback}
When feedback is available to the transmitting users, the sum-rate in Proposition~\ref{C ZF uniform users exact} is 
\begin{align*}
&\mathcal{R}_{fb}(u_k) = \sum_{j=1}^{\infty}\frac{k^{j}e^{-k}}{j!} \sum_{i=1}^{\min\prnt{j,r}} \E\brkt{\log\Big(\left.1 + P \Arrowvert \bV_i\blue{(j)} \bh_i \Arrowvert^2 \Big) \right|  \Arrowvert \bh_i \Arrowvert^2 > u_k} &
\end{align*}
where $u_k$ is a threshold\blue{, set such that the number of users that exceed it on average is $k$}.
\end{corollary}
\blue{The dependence of $\bV_i(j)$ on the number of active users $j$ is even more important in Corollary~\ref{coro: delta sum-rate with feedback}, as it reveals the gain in sum-rate that is obtained from avoiding collision slots, compared to scheduling without feedback. In particular, note that while for $\min\{j,r\} = j$ the statistics of $\bV_i$ depends on $j$, for $\min\{j,r\} = r$, the statistics of $\bV_i$ is the same for all $j$ since it depends only on $r$. Therefore, the additional rate gained by the scheduler, $\mathcal{R}_{\Delta}$, can be quantified, and is exactly}

\red{\mbox{Corollary~\ref{coro: delta sum-rate with feedback}} implicitly states that the gain in sum-rate that is originate from collision slots, compared to scheduling without feedback, is exactly}\label{commentr28}
\begin{eqnarray*}
\mathcal{R}_{\Delta}(u_k) &=& \sum_{j=r+1}^{\infty}\frac{k^{j}e^{-k}}{j!} \sum_{i=1}^{r}  \E\left[\left.\log\Big(1 + P \Arrowvert \bV_i\blue{(r)} \bh_i \Arrowvert^2 \Big) \right|  \Arrowvert \bh_i \Arrowvert^2 > u_k \right]\\
&=& \blue{\sum_{j=r+1}^{\infty}\frac{k^{j}e^{-k}}{j!} r \E\left[\left.\log\Big(1 + P \Arrowvert \bV_{1}(r) \bh_{1} \Arrowvert^2 \Big) \right|  \Arrowvert \bh_1 \Arrowvert^2 > u_k \right]}\\
&=& \blue{r \E\left[\left.\log\Big(1 + P \Arrowvert \bV_1(r) \bh_1 \Arrowvert^2 \Big) \right|  \Arrowvert \bh_1 \Arrowvert^2 > u_k \right] \sum_{j=r+1}^{\infty}\frac{k^{j}e^{-k}}{j!}} .
\end{eqnarray*}
The upper and lower bounds in Lemma~\ref{C ZF uniform users upper} and Lemma~\ref{C ZF uniform users lower}, respectively, are also applicable to this gain, in addition to a smaller delay in the decoding process.


Extensive simulations were conducted to compare the analytical bounds derived above to real world situations with a finite number of users. The setting we explore is the following: Each block, each of the K users has a new, independent channel vector. As mentioned, these vectors are drawn using a Gaussian distribution. Our algorithm has a parameter $k$, and sets a threshold such that on average $k$ out of $K$ will exceed the threshold for each block. The exceeding users transmit, they are detected using ZF or MMSE, and so on. Now, clearly, as the channel vectors of the users where chosen at random, the resulting sum-rate under this procedure is also random. This is the randomness we explore. Specifically, in the analytical results, we give upper and lower bounds on the \emph{expected} sum-rate, and in Figures~\ref{fig:mac_sim248}-\ref{fig:mac_sim24830} we compare the bounds to simulation results, when vectors are chosen at random, users are "selected" using the algorithm, and the sum-rate is computed. In particular, a random channel was drawn for each user from the complex Gaussian distribution, then, a group of users were selected according to the {\scshape TB-Channel-Access} scheme, and the resulting sum-rate was calculated. This experiment was repeated $1000$ times, such that the average sum-rate is a good approximation for the expected sum-rate (assuming ergodicity).
In Figure~\ref{fig:mac_sim248}, we compare the bounds on the \RED{expected achievable} \BLUE{ergodic} sum-rate of the {\scshape TB-Channel-Access} scheme under ZF-receiver (i.e., Lemma~\ref{C ZF uniform users upper} and Lemma~\ref{C ZF uniform users lower}), to the simulation results, for a different threshold values $u_k$. The bounds and simulation are for $K=300$ users, with $2,4$ and $8$ receiving antennas at the BS. The tightness of the upper bound is clearly visible. While the lower bound is looser, it still gives the correct behavior as a function of the number of users passing the threshold on average, $k$. Indeed, it is clear from the figure that a key factor affecting the system performance is the number of users passing the threshold. This distribution gives the graphs their Poisson-like shape.

To see that the trend holds even for a relatively small number of users, Figure~\ref{fig:mac_sim24830} includes the same plots for $K=30$. Note that even for $30$ users the algorithm manages to achieve a significant rate (compared, e.g., to the one achieved with $300$ users). This means the essence of the multi-user diversity is exploited by the algorithm even for a relatively small number of users. Note, however, that as $r$ approaches $K$ then the accuracy of the bounds decreases.

The results in \Cref{fig:mac_sim248} and \Cref{fig:mac_sim24830} give the expected sum-rate as a function of $k$. The \RED{capacity} \BLUE{rate} distribution, compared to different algorithms, will be given in the subsequent sections.
\begin{figure}
	\centering
		\includegraphics[scale=0.35,angle=-90]{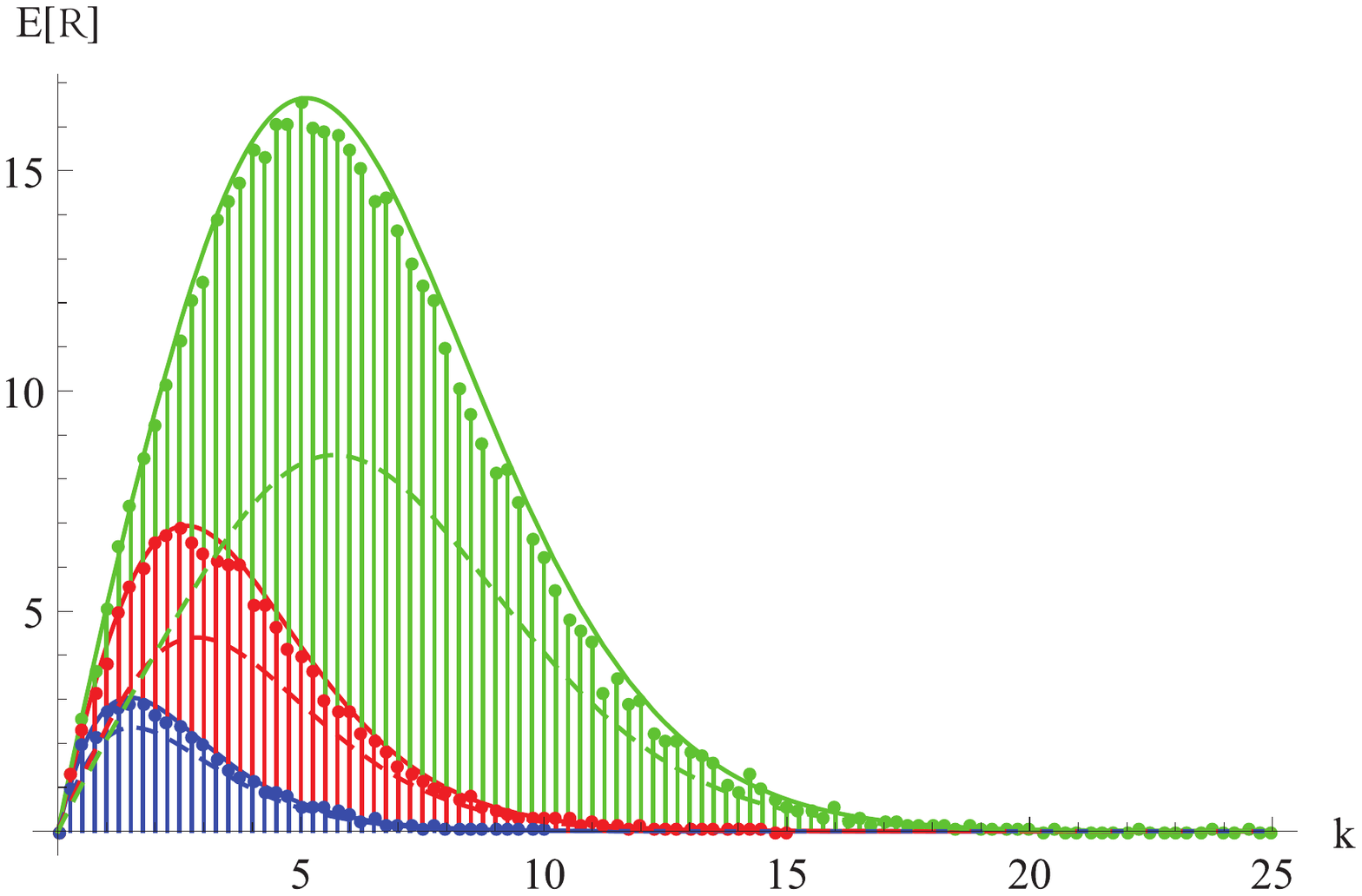}
	\caption{\RED{Expected achievable} \BLUE{Ergodic} sum-rate under the {\scshape TB-Channel-Access} scheme and $K=300$. Bars are simulation results, while the solid and dashed lines represent the upper and lower bounds, respectively. The threshold $u_k$ is set such that $k$ users exceed it on average. The green (upper), red (middle) and blue (lower) lobes are for $r=8,4$ and $2$ receive antennas, respectively. Note that the optimal $k$ is smaller than $r$. This is since it is better to aim at slightly less than $r$ users passing, as the performance deteriorates significantly for more than $r$ users transmitting concurrently.}
	\label{fig:mac_sim248}
\end{figure}

\begin{figure}
	\centering
		\includegraphics[scale=0.35,angle=-90]{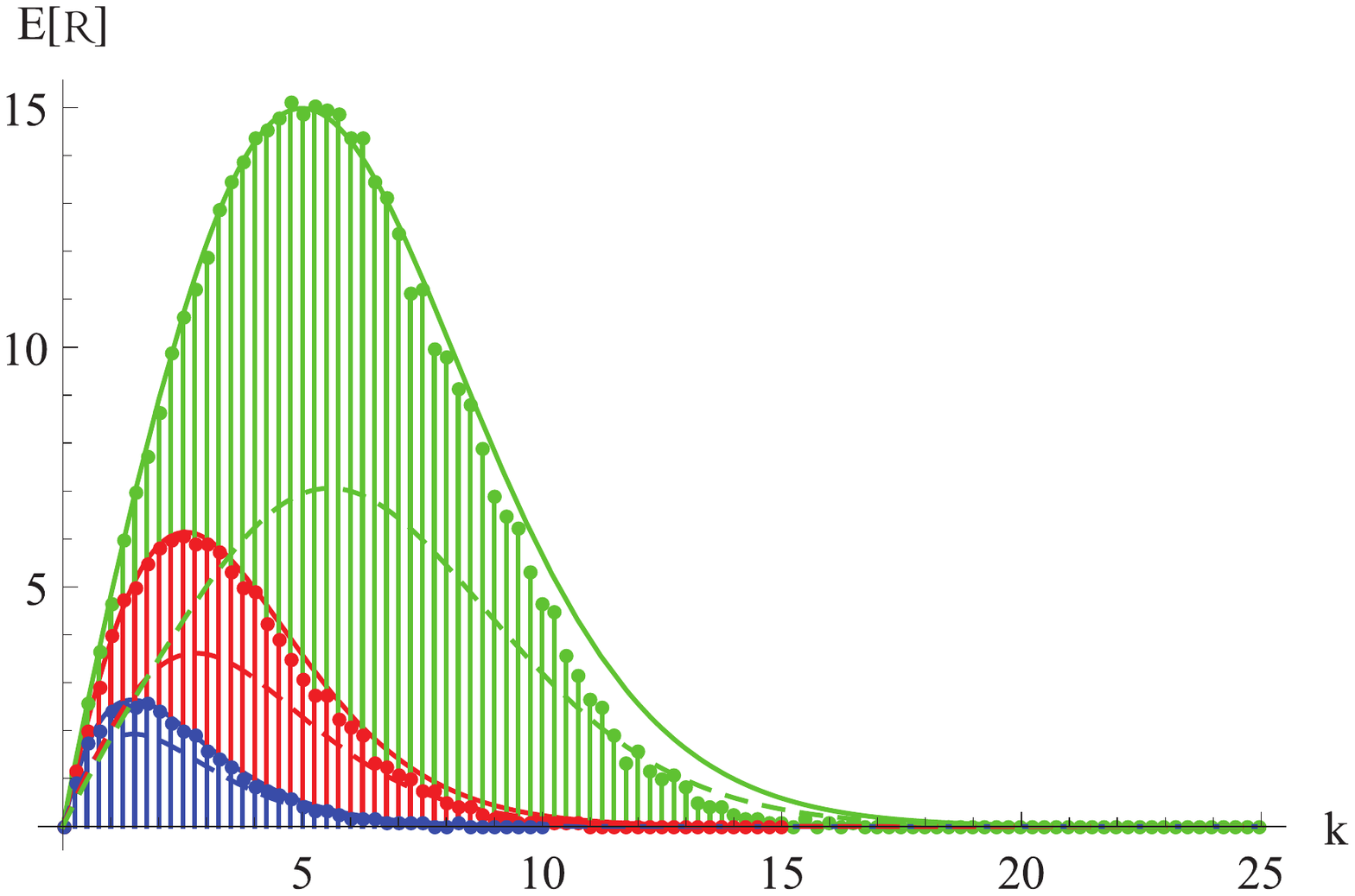}
	\caption{Expected achievable sum-rate under the {\scshape TB-Channel-Access} scheme  and $K=30$. Bars are simulation results, while the solid and dashed lines represent the upper and lower bounds, respectively. The threshold $u_k$ is set such that $k$ users exceed it on average. The green (upper), red (middle) and blue (lower) lobes are for $r=8,4$ and $2$ receive antennas, respectively. Note that the essence of the analysis holds for even a relatively small number of users as $30$, as the bounds still depict the key elements affecting the performance.}
	\label{fig:mac_sim24830}
\end{figure}

\section{MMSE receiver}\label{sec. MMSE rec}
The ZF receiver discussed thus far is on the one hand simple enough to facilitate rigorous analysis, yet, as shown in the previous sections, powerful enough in the sense that with intelligent user selection (in this paper, distributed) can achieve the optimal scaling laws. Still, this is not the optimal linear receiver.
In this section, we explore the scaling laws of the \RED{expected achievable} \BLUE{ergodic} sum-rate of the MMSE receiver.
\BLUE{Note that, similar to the ZF detection, also in MMSE detection, the $SINR$ of user $i$ depends on all other above-threshold users, i.e., on the whole channel matrix $\bH$, via matrix $\bR$. } 

As mentioned, in this case we let $\bH$ denote a matrix whose columns are the channel vectors of the \emph{transmitting users}. That is, when using the {\scshape TB-Channel-Access} algorithm, vectors with norm greater than the  threshold. $\bH_{(-i)}$ is the matrix $\bH$ with its $i_{th}$ column removed. Under these definitions, the SINR seen at the $i$th stream was given in \eqref{sinr in mmse}.
Let $\mathcal{S}$ denote the set of channels with norm greater than a threshold. Then, the expected achievable sum-rate under the {\scshape TB-Channel-Access} algorithm in \Cref{sec. dist alg} is as follows.
\begin{proposition}\label{C MMSE uniform users}\label{prop. C mmse}
For $K \gg r$, the \RED{expected achievable} \BLUE{ergodic} sum-rate with MMSE detection is
\begin{align*}
	\mathcal{R}(u_k) =& \sum_{j=1}^{r}\frac{k^{j}\e^{-k}}{j!} \sum_{i=1}^j \E\left[\left.\log\left(1 + P \bh_i^{\dagger}\bR \bh_i \right)\right| \Vert \bh_s \Vert^2 > u_k, \forall s \in \mathcal{S} \right] + O\left(\frac{\log \log K}{K}\right)
\end{align*}
where $u_k$ is a threshold\blue{, set such that the number of users that exceed it on average is $k$}.
\end{proposition}
This expected sum-rate should be optimized over $k$. The proof is similar to the ZF setting. The only change is in the SNR seen by the users, as reflected by the term within the $\log$.
\off{
\begin{proof}
Similar to the ZF setting, let $J$ be a random variable representing the number of users exceeding $u_k$. Clearly, $\text{E}\{C(u_k)\} = \text{E}_J \text{E}\{C(u_k)|J\}$. That is, the expected sum-rate in a slot is determined by the number of users who exceed the threshold and the sum of capacities these users see, given that they exceeded the threshold. In this paper, we restrict our attention to $j \leq r$ transmitting users. Thus, if more than $r$ user begin transmission,  we say that a collision occurred, and the sum-rate in that slot is considered as zero. However, note that in this scheme only users above the threshold may interfere with each other.

When the users are i.i.d., the probability of $j$ threshold exceedances follows the binomial distribution with probability $p= k/K$ to exceed the threshold $u_k$. Since we consider large $K$ and small values of $p$, the number of users to exceed threshold can be approximated by the Poisson distribution with an approximation error in the order of $1/K$. Hence the Poisson expression for $J=j$. The sum-rate given $J=j$ is the standard capacity expression under MMSE receiver as users are decoded separately.
\end{proof}
}
\begin{remark}\label{scaling remark}
\blue{When $r$ strongest users are selected out of $K$ users, and $K$ is large, the channel gain (of these strong users) scales as $O(\log K)$. It is interesting to see how this affects the SNR in both ZF and MMSE. Rewrite $\bh_i$ as $\alpha \bw_i$, where $\bw_i = \frac{\bh_i}{\Vert \bh_i \Vert}$. Thus, $\alpha^2$ represents the scaling of the channel norm $\Vert \bh_i \Vert ^2$.

When using ZF, the effective SNR seen by a user is $P\Vert \bV_i \bh_i \Vert^2 = \alpha^2 P\Vert \bV_i \bw_i \Vert^2$. As the matrix $\bV_i$ is unitary, the scaling of the SNR with $\alpha$ is clear. In the MMSE case, however, the analysis is intricate. The SNR is $P \bh_i^{\dagger}\bR \bh_i = \alpha^2 P \bw_i^{\dagger}\bR \bw_i$. However, at first sight, $\bR = (\blue{P}\bH_{(-i)}\bH_{(-i)}^\dagger+I)^{-1}$ appears to contribute a $\frac{1}{\alpha^2}$ factor, as the eigenvalues of $\bH_{(-i)}\bH_{(-i)}^\dagger$ increase like $\alpha^2$. Fortunately, $\bH_{(-i)}\bH_{(-i)}^\dagger$ is \emph{not full rank}, hence has at least one zero eigenvalue, which does not increase with $\alpha$. As a result, $\bR$ has at least one eigenvalue which remains constant and does not scale down when $\alpha$ increases, and therefore the SNR under MMSE scales up like $\alpha^2$ as well.}
\end{remark}

To evaluate the \RED{expected achievable} \BLUE{ergodic} sum-rate in \Cref{prop. C mmse}, the characteristics of the random variables $\bh_i^{\dagger} \bR \bh_i$  should be understood, especially when the number of transmitting users is random. Further, the influence of the norms  $\Vert \bh_s \Vert^2$ on $\bh_i^{\dagger} \bR \bh_i$ should be evaluated.
The key technical challenge, however, is due to the norm condition \emph{inducing dependence on the matrix elements}, hence the random matrix theory usually used in the MIMO literature does not hold. Part of the contribution in this section, is by bringing new tools to tackle this problem.

Note that in this section, regular type letters  represent random variables as well as scalar variables. The difference will be clear from the context.
As a first tool to handle the dependence within the matrix entries, we start with the following lemma:
\begin{lemma}\label{claim. entries var of above thr vec}
Assume the norms $\Vert \bh_s \Vert^2, \forall s \in \mathcal{S}$ are above a given threshold $u_k$. Then the following properties hold:
\begin{itemize}
	\item[(i)]  The entries of the channel vector remain zero mean.
	\item[(ii)] The variance of each entry in the vector  scales. In particular, is equals to
\begin{equation*}
\E\left[\left. \vert h_{i,n} \vert^2 \right| \Vert \bh_i \Vert^2 > u_k \right] = (u_k + a_{\{K,r\}})/r.
\end{equation*}
\item[(iii)] The vector elements remain uncorrelated in pairs.
\end{itemize}
\end{lemma}
Note that, to begin with, the entries of $\bH$ are i.i.d. The lemma states that conditioned on exceeding a threshold, while not i.i.d., they sustain the zero correlation. This property will be useful throughout the remainder of this paper.
The proof of \Cref{claim. entries var of above thr vec} is deferred to \Cref{proof:above_thr_vec_entries}.
\subsection{{Threshold-based MMSE Upper Bound}}
Now we are ready to derive the scaling-law of the MMSE receiver, using the following upper and lower bounds on the {\scshape TB-Channel-Access} scheme \RED{expected  achievable} \BLUE{ergodic} sum-rate.   
\begin{lemma}\label{lem. C mmse upper}
The \RED{expected  achievable} \BLUE{ergodic} sum-rate of Algorithm {\scshape TB-Channel-Access} with MMSE detection satisfies the following upper bound.
\begin{multline*}
	\mathcal{R}(u_k) \leq \sum_{j=1}^{r}\frac{k^{j}\e^{-k}}{j!} j \\ \log\left(1 + P\bigg( 1 - \frac{(j-1) u_k^2}{r\left((1+\frac{j}{r})(u_k + a_{\{K,r\}})^2 +a_{\{K,r\}}^2 +(a_{\{K,r\}}+u_k)/\blue{P}\right)}\bigg) (u_k + a_{\{K,r\}})  \right). 
\end{multline*}
where $a_{\{K,r\}}$ is given by (\ref{eqn: a_n normalized}) and $u_k$ is a threshold\blue{, set such that the number of users that exceed it on average is $k$}.
\end{lemma}
Before we prove the lemma, it is interesting to compare the scaling law under MMSE decoding to that achieved with ZF decoding. In both \Cref{C ZF uniform users upper} and \Cref{lem. C mmse upper}, the rate seen by each user is approximately $\log(1+cPu_k)$, for some constant $c$. Asymptotically, it follows that $0<c<1$, and in both cases, the rate gain comes from the threshold value $u_k$, which is, as mentioned, $O(\log K)$. This gives the growth rate of $\log \log K$ per user. However, note that while $c = \frac{r-j+1}{r}$ in \Cref{C ZF uniform users upper}, for large enough $u_k$ the gain in \Cref{lem. C mmse upper} is asymptotically $c=\frac{r+1}{r+j}$, that is, a larger gain for any $j \geq 1$. This is not surprising, as an MMSE detector does give a better power gain, but does not improve the already optimal scaling law.
\begin{proof}[Proof (\Cref{lem. C mmse upper})]
The rate seen by user $i$ is bounded by:
\begin{flalign}\label{eq: mmse upper bound}
\mathcal{R}_i (u_k) &\stackrel{(a)}{\leq} \log\left(1+ P \text{E}\left[\left. \bh_i^{\dagger}\bR \bh_i\right| \Vert \bh_s \Vert^2 > u_k, \forall s \in \mathcal{S} \right]  \right)\nonumber\\
&\stackrel{(b)}{=} \log\left(1+ P \text{E}\left[\left. \sum_{n=1}^r\sum_{m=1}^{r} h_{i,m}^{\ast}h_{i,n}\left[\bR\right]_{mn}\right| \Vert \bh_s \Vert^2 > u_k, \forall s \in \mathcal{S} \right]  \right)\nonumber\\
&\stackrel{(c)}{=} \log\left(1+ P \sum_{n=1}^r\sum_{m=1}^{r} \text{E}\left[  h_{i,m}^{\ast}h_{i,n}\Big| \Vert \bh_i \Vert^2 > u_k \right] \text{E}\left[ \left[\bR\right]_{mn} \Big| \Vert \bh_s \Vert^2 > u_k, \forall s \in \mathcal{S}\right]  \right) \nonumber\\
&\stackrel{(d)}{=} \log\left(1+ P \sum_{n=1}^{r} \text{E}\left[ \vert h_{i,n} \vert^2 \Big| \Vert \bh_i \Vert^2 > u_k \right] \text{E}\left[ \left[\bR\right]_{nn} \Big| \Vert \bh_s \Vert^2 > u_k, \forall s \in \mathcal{S}\right]  \right) \nonumber\\
&\stackrel{(e)}{=} \log\left(1+ P  \frac{(u_k + a_{\{K,r\}})}{r}\text{E}\left[ \text{tr}\left(\bR\right) \Big| \Vert \bh_s \Vert^2 > u_k, \forall s \in \mathcal{S}\right]  \right)
\end{flalign}
In the above chain of equalities, (a) follows from Jensen inequality. (b) is the equivalent quadric form of $\bh_i^{\dagger}\bR \bh_i$. (c) \blue{is since $R$ is a matrix involving only random variables in $\bH_{(-i)}$, that is, excluding $\bh_i$, hence the random variables in $\bR$ and $\bh_i$ are independent.}\red{the vector $\bh_i$ is independent in the column vectors of the matrix $\bR$.} (d) is since by \Cref{claim. entries var of above thr vec} part (iii), the elements of the vector $\bh_i$ are uncorrelated. (e) follows from \Cref{claim. entries var of above thr vec} part (ii).

 Now, to address $\text{tr}\left(\bR\right)$, conditioned on $\Vert \bh_s \Vert^2 > u_k, \forall s \in \mathcal{S}$,  we first use the following upper bound on the trace of the inverse of a matrix.
 For a symmetric $r \times r$ positive definite matrix  A \cite{bai1996bounds},
\begin{equation}
\text{tr}\left(A^{-1}\right) \leq \begin{pmatrix} \text{tr}(A) & r \end{pmatrix} \begin{pmatrix} \text{tr}(A^{\dagger}A) &  \text{tr}(A) \\ \lambda_{\min}^2(A) & \lambda_{\min}(A) \end{pmatrix}^{-1} \begin{pmatrix} r \\ 1 \end{pmatrix},
\label{eq: trace inequality}
\end{equation}
where $\lambda_{\min}(A)$ is the smallest eigenvalue of the matrix $A$.
In our case, take $A^{-1} = \bR$. To obtain $\lambda_{\min}(\bR^{-1})$, we recall that $\bR^{-1} = \left(\blue{P}\bH_{(-i)}\bH_{(-i)}^{\dagger} + I\right)$ and that $\bH_{(-i)}\bH_{(-i)}^{\dagger}$ is an $r \times r$ matrix with rank $j-1$, with $j \leq r$, thus, by the eigenvalues decomposition theorem, $\bH_{(-i)}\bH_{(-i)}^{\dagger}$ has $r-j+1$ of its eigenvalues equal to zero. Hence, the smallest eigenvalue of $\bR^{-1}$ is $\lambda_{\min}(\bR^{-1}) = 1$.
Thus, using (\ref{eq: trace inequality}), it follows that
\begin{eqnarray*}
\text{tr}(\bR) &\leq& \frac{r \left[\text{tr}\left(\bR^{-1}\right) + \text{tr}\left(\bR^{-\dagger}\bR^{-1}\right) \right] - \text{tr}^2\left(\bR^{-1}\right) - r^2}{\text{tr}\left(\bR^{-\dagger}\bR^{-1}\right) - \text{tr}\left(\bR^{-1}\right)}\\
&=& \frac{r\left[\blue{P^2}\text{tr}\left(\bH_{(-i)}\bH_{(-i)}^{\dagger} \bH_{(-i)}\bH_{(-i)}^{\dagger}\right) + \blue{P}\text{tr}\left(\bH_{(-i)}\bH_{(-i)}^{\dagger}\right) \right] - \blue{P^2}\text{tr}^2\left(\bH_{(-i)}\bH_{(-i)}^{\dagger}\right)}{\blue{P^2}\text{tr}\left(\bH_{(-i)}\bH_{(-i)}^{\dagger} \bH_{(-i)}\bH_{(-i)}^{\dagger}\right) + \blue{P}\text{tr}\left(\bH_{(-i)}\bH_{(-i)}^{\dagger}\right)}\\
&=& r - \frac{\text{tr}^2\left(\bH_{(-i)}^{\dagger}\bH_{(-i)}\right)}{\text{tr}\left(\bH_{(-i)}^{\dagger}\bH_{(-i)}\bH_{(-i)}^{\dagger} \bH_{(-i)}\right) + \blue{\frac{1}{P}}\text{tr}\left(\bH_{(-i)}^{\dagger}\bH_{(-i)}\right)}.
\label{eq: inverse trace bound}
\end{eqnarray*}
Note that the expression  in~(\ref{eq: mmse upper bound}) had a power loss of factor $r$, that is, $P$ is divided by $r$. However, it is multiplied by $E\left[\text{tr}(\bR) \Big| \Vert \bh_s \Vert^2 > u_k\right]$. In the last equality above we actually see that in this multiplication we gain back a factor of the form $(r- \zeta)$. This suggests that the MMSE detector is superior to the ZF detector in terms of power gain, as the expression in the ZF detector performance does not have this multiplicative factor (this in consistent with, e.g., \cite{tse2005fundamentals}). We will later evaluate the value of $\zeta$ explicitly.
Furthermore, note that $\text{tr}\left(\bR^{-\dagger}\bR^{-1}\right)$ is actually the squared Frobenius norm of $\bR^{-1}$, which is equal to the sum of squares of it entries.
   Thus, we have,
\begin{flalign*}
&\text{E}\left[\left. \text{tr}(\bR)\right| \Vert \bh_s \Vert^2 > u_k, \forall s \in \mathcal{S} \right] & \\
&\leq r - \text{E}\left[\left.\frac{\text{tr}^2\left(\bH_{(-i)}^{\dagger}\bH_{(-i)}\right)}{\text{tr}\left(\bH_{(-i)}^{\dagger}\bH_{(-i)}\bH_{(-i)}^{\dagger} \bH_{(-i)}\right) + \blue{\frac{1}{P}}\text{tr}\left(\bH_{(-i)}^{\dagger}\bH_{(-i)}\right)}\right| \Vert \bh_s \Vert^2 > u_k, \forall s \in \mathcal{S} \right] &\\
&= r - \text{E}\left[\left.\frac{\left(\sum_{\substack{n=1 \\ n \neq i}}^{j}\Vert\bh_n\Vert^2\right)^2}{\sum_{\substack{n=1 \\ n \neq i}}^{j}\left(\Vert \bh_n \Vert ^2\right)^2 + 2\sum_{\substack{n=1 \\ n \neq i}}^{j}\sum_{\substack{m=n+1 \\ m \neq i}}^{j}\vert \langle\bh_n,\bh_m \rangle \vert^2 +\blue{\frac{1}{P}}\sum_{\substack{n=1 \\ n \neq i}}^{j}\Vert\bh_n\Vert^2} \right| \Vert \bh_s \Vert^2 > u_k, \forall s \in \mathcal{S}\right]&\\
&=  r - \text{E}\left[\frac{\left(\sum_{\substack{n=1 \\ n \neq i}}^{j}\Vert\bh_n\Vert^2\right)^2}{\sum_{\substack{n=1 \\ n \neq i}}^{j}\left(\Vert \bh_n \Vert ^2\right)^2 + 2\sum_{\substack{n=1 \\ n \neq i}}^{j}\sum_{\substack{m=n+1 \\ m \neq i}}^{j} \frac{\vert \langle\bh_n,\bh_m \rangle \vert^2}{\Vert\bh_n\Vert^2 \Vert\bh_m\Vert^2}\Vert\bh_n\Vert^2 \Vert\bh_m\Vert^2 +\blue{\frac{1}{P}}\sum_{\substack{n=1 \\ n \neq i}}^{j}\Vert\bh_n\Vert^2} \right. &\\
&  \quad \hspace{2cm}\left.  \Bigg| \Vert \bh_s \Vert^2 > u_k, \forall s \in \mathcal{S}\right]&
\end{flalign*}
Note that the numerator inside the expectation is at least $(j-1)^2 u_k^2$. Thus,
\begin{flalign*}
&\text{E}\left[\left. \text{tr}(\bR)\right| \Vert \bh_s \Vert^2 > u_k, \forall s \in \mathcal{S} \right] & \\
&\leq r +  \text{E}\left[\frac{-(j-1)^2 u_k^2}{\sum_{\substack{n=1 \\ n \neq i}}^{j}\left(\Vert \bh_n \Vert ^2\right)^2 + 2\sum_{\substack{n=1 \\ n \neq i}}^{j}\sum_{\substack{m=n+1 \\ m \neq i}}^{j} \frac{\vert \langle\bh_n,\bh_m \rangle \vert^2}{\Vert\bh_n\Vert^2 \Vert\bh_m\Vert^2}\Vert\bh_n\Vert^2 \Vert\bh_m\Vert^2 +\blue{\frac{1}{P}}\sum_{\substack{n=1 \\ n \neq i}}^{j}\Vert\bh_n\Vert^2} \right. &\\
& \quad \hspace{2cm} \left. \Bigg| \Vert \bh_s \Vert^2 > u_k, \forall s \in \mathcal{S}\right]&
\end{flalign*}
Now, we have a concave function of the form $-1/x$ inside the expectation. Accordingly, using Jensen inequality for concave functions, we obtain that

\begin{flalign*}
&\text{E}\left[\left. \text{tr}(\bR)\right| \Vert \bh_s \Vert^2 > u_k, \forall s \in \mathcal{S} \right] & \\
&\leq r - (j-1)^2 u_k^2 \left(\text{E}\left[ \sum_{\substack{n=1 \\ n \neq i}}^{j}\left(\Vert \bh_n \Vert ^2\right)^2  + 2\sum_{\substack{n=1 \\ n \neq i}}^{j}\sum_{\substack{m=n+1 \\ m \neq i}}^{j} \frac{\vert \langle\bh_n,\bh_m \rangle \vert^2}{\Vert\bh_n\Vert^2 \Vert\bh_m\Vert^2}\Vert\bh_n\Vert^2 \Vert\bh_m\Vert^2 \right. \right. &\\
& \quad \hspace{3cm}  + \left. \left. \left. \blue{\frac{1}{P}}\sum_{\substack{n=1 \\ n \neq i}}^{j}\Vert\bh_n\Vert^2 \right| \Vert \bh_s \Vert^2 > u_k, \forall s \in \mathcal{S} \right] \right)^{-1}&\\
&= r - (j-1)^2 u_k^2 \left(\sum_{\substack{n=1 \\ n \neq i}}^{j}\text{E}\left[\left. \left(\Vert \bh_n \Vert ^2\right)^2 \right| \Vert \bh_n \Vert^2 > u_k \right] \right. &\\
& \quad \hspace{0.5cm} + 2\sum_{\substack{n=1 \\ n \neq i}}^{j}\sum_{\substack{m=n+1 \\ m \neq i}}^{j} \text{E}\left[\left. \frac{\vert \langle\bh_n,\bh_m \rangle \vert^2}{\Vert\bh_n\Vert^2 \Vert\bh_m\Vert^2} \right| \Vert \bh_s \Vert^2 > u_k, \forall s \in \mathcal{S} \right] &\\
&\quad \hspace{0.5cm} \cdot \text{E}\left[\left.\Vert\bh_n\Vert^2 \right| \Vert \bh_n \Vert^2 > u_k\right] \text{E}\left[\left.  \Vert\bh_m\Vert^2 \right|\Vert \bh_m \Vert^2 > u_k\right] +  \left. \blue{\frac{1}{P}}\sum_{\substack{n=1 \\ n \neq i}}^{j}\text{E}\left[\left. \Vert\bh_n\Vert^2 \right| \Vert \bh_n \Vert^2 > u_k \right] \right)^{-1}&\\
\end{flalign*}
where the last equality follows from the linearity of the expectation and  since  $\frac{\langle\bh_n,\bh_m \rangle^2}{\Vert\bh_n\Vert^2 \Vert\bh_m\Vert^2}$, which is the angle between two independent vectors, $\bh_m$ and $\bh_n$, is independent of the vectors' norms. Moreover, we can now remove the conditioning on the norm when taking the expectation of the angle.
Accordingly, an upper bound can be evaluated by computing the expectations in the denominator as follows.
\begin{flalign*}
 &r - (j-1)^2 u_k^2 \left(\sum_{\substack{n=1 \\ n \neq i}}^{j}\text{E}\left[\left. \left(\Vert \bh_n \Vert ^2\right)^2 \right| \Vert \bh_n \Vert^2 > u_k \right] \right. &\\
& \quad \hspace{0.5cm} + 2\sum_{\substack{n=1 \\ n \neq i}}^{j}\sum_{\substack{m=n+1 \\ m \neq i}}^{j} \text{E}\left[ \frac{\vert \langle\bh_n,\bh_m \rangle \vert^2}{\Vert\bh_n\Vert^2 \Vert\bh_m\Vert^2}\right]\text{E}\left[\left.\Vert\bh_n\Vert^2 \right| \Vert \bh_n \Vert^2 > u_k\right] \text{E}\left[\left.  \Vert\bh_m\Vert^2 \right|\Vert \bh_m \Vert^2 > u_k\right] &\\
& \quad \hspace{3cm} +  \left. \blue{\frac{1}{P}} \sum_{\substack{n=1 \\ n \neq i}}^{j}\text{E}\left[\left. \Vert\bh_n\Vert^2 \right| \Vert \bh_n \Vert^2 > u_k \right] \right)^{-1}&\\
& \stackrel{(a)}{=}  r - (j-1)^2 u_k^2 &\\
& \quad \hspace{0.5cm} \cdot \left((j-1)(u_k^2 +2a_{\{K,r\}}u_k + 2a_{\{K,r\}}^2) + j(j-1) \frac{1}{r}(u_k + a_{\{K,r\}})^2 +\frac{(j-1)}{\blue{P}}(u_k +a_{\{K,r\}}) \right)^{-1}&\\
& =  r - \frac{(j-1) u_k^2}{\left(1+\frac{j}{r}\right)(u_k + a_{\{K,r\}})^2 +a_{\{K,r\}}^2 + (a_{\{K,r\}}+u_k)/\blue{P}},
\end{flalign*}
where (a) follows since the limit distribution of the channel norm tail is exponentially distributed with rate parameter $1/a_{\{K,r\}}$ given that it is above high threshold. Namely, $ \Vert \bh_n \Vert ^2 \Big| \Vert \bh_n \Vert ^2 >u_k \sim Exp(1/a_{\{K,r\}})$, \cite{leadbetter1983}.  Thus, $\text{E}\left[\left. \left(\Vert \bh_n \Vert ^2\right)^2 \right| \Vert \bh \Vert^2 > u_k \right]$ can be interpreted as a second moment of exponential random variable. Similar to \Cref{sec. dist alg}, by \cite[Lemma 3.2]{jagannathan2006efficient}, the angle has the same distribution as that of the minimum of $r-1$ independent uniform $[0,1]$ random variables (i.e., with CDF $1-(1-\alpha)^{r-1}$, $0 \leq \alpha \leq 1$). The  expectations of the norms, again, follow from the EVT. Details are in \Cref{coro: expected capacity above est by GEV thr}, \Cref{sec: tail dist}.
\end{proof}

\subsection{Threshold-based MMSE Lower Bound}
To completely characterize the performance of the suggested {\scshape TB-Channel-Access} algorithm under MMSE decoding, we proceed to derive a corresponding lower bound. To this end, the following lemmas will be useful. The proofs are deferred to \Cref{proof:R_unitarily_invariant}.

\begin{lemma}\label{claim. R is unitarily invariant}
Let $\bH$ be a complex Gaussian matrix with i.i.d.\ entries. Then,
\begin{itemize}
\item[(i)] $\bR^{-1} = \left(\blue{P}\bH \bH^{\dagger} + I \right)$ is unitary invariant.
In particular, $\bR$ can be decomposed as $\bU \pmb{\Lambda} \bU^{\dagger}$, where $\bU$ is a unitary matrix, independent of the diagonal matrix $\pmb{\Lambda}$.
\item[(ii)] The property above holds conditioned on all norms of the columns $\bh_i$ of $\bH$ being above a threshold $u$. In particular, conditioned on $\Vert \bh_i \Vert^2 > u$, $U \bh_i$ and $\bh_i$ have the same distribution.
\end{itemize}
\end{lemma}
Let $\Gamma(s,x) = \int_{x}^{\infty} \tau^{s-1}\e^{-\tau} \mathrm{d}\tau$ denote the upper incomplete Gamma function.
\begin{claim}[E.g., \cite{GammaBound}]\label{claim. lower bound on upper gamma s geq 2}
For $s \geq 2$, we have $\Gamma(s,x) \geq \e^{-x}(1+x)^{s-1}$.
\end{claim}
\begin{lemma}\label{claim. lower bound on expectation}
Let $\bx \sim \chi^2_{2(r-j+1)}$ and let $\by \sim \chi^2_{2(j-1)}$, independent of $\bx$, for some integer $1<j\leq r$. Then,
\begin{multline*}
 \E\brkt{\left.\log(1+\bx)\right|\bx +\by > u}
 >  \left(\frac{\Gamma(r) \Gamma(r-j+1,u/2)}{\Gamma(r,u/2)\Gamma(r-j+1)} + \frac{\e^{-u/2}u^{r-1}}{2^{r-1}\Gamma(r,u/2)}\right)\log(u) \\
    \quad \quad  + \frac{\e^{-u/2}u^{r-1}}{2^{r-1}\Gamma(r,u/2)} \left( \psi(r-j+1) - \psi(r) \right) + e \frac{\Gamma(r) \Gamma(r-j,1+u/2)}{\Gamma(r,u/2)\Gamma(r-j+1)},
\end{multline*}
where $\psi(x) = \frac{d}{dx}\ln \Gamma(x)$ is the Digamma function.
\end{lemma}
To understand \Cref{claim. lower bound on expectation}, note that $\Gamma(s,x)=(s-1)\Gamma(s-1,x)+x^{s-1}\e^{-x}$. Hence, it is not hard to show that
$$\lim_{u \to \infty} \frac{\Gamma(r) \Gamma(r-j+1,u/2)}{\Gamma(r,u/2)\Gamma(r-j+1)} = 1.$$
Moreover, using  \Cref{claim. lower bound on upper gamma s geq 2} above, one can also show that
$$0\leq \frac{\e^{-u/2}u^{r-1}}{2^{r-1}\Gamma(r,u/2)} \leq 1.$$
Hence, the importance of \Cref{claim. lower bound on expectation} is in showing that $E[\log(1+\bx)|\bx +\by > u_k]=\Theta(\log u_k)$. Similar to the upper bound, since $u_k = \Theta(\log K)$, the scaling law of $\Theta(\log\log K)$ will result.

To conclude, a lower bound on the performance of the {\scshape TB-Channel-Access} algorithm with MMSE receiver is as follows. For simplicity of the presentation, we assume $P=1$.
\begin{lemma}\label{lem. C mmse lower}
For sufficiently large $K$, the \RED{expected achievable} \BLUE{ergodic} sum-rate under the {\scshape TB-Channel-Access} algorithm and MMSE detection satisfies the following lower bound.
\begin{multline*}
	\mathcal{R}(u_k)
	 > \sum_{j=2}^{r}\frac{k^{j}\e^{-k}}{j!} j \left[\left(\frac{\Gamma(r) \Gamma(r-j+1,u_k/2)}{\Gamma(r,u_k/2)\Gamma(r-j+1)} + \e^{-u_k/2}\frac{u_k^{r-1}}{2^{r-1}\Gamma(r,u_k/2)}\right)\log(u_k) \right. \\
    \quad \quad \left. + \e^{-u_k/2}\frac{u_k^{r-1}}{2^{r-1}\Gamma(r,u_k/2)} \left( \psi(r-j+1) - \psi(r) \right) + e \frac{\Gamma(r) \Gamma(r-j,1+u_k/2)}{\Gamma(r,u_k/2)\Gamma(r-j+1)} \right].
\end{multline*}
\end{lemma}
\begin{proof}
We assume $1<j\leq r$ users begin their transmission simultaneously and let $\mathcal{S}$ denote the set of channels with norm greater than threshold. Note that when only one user passed the threshold, while its rate is easy to compute, as it is the single-user MISO capacity, within this bound it is negligible as it is multiplied by the probability (under the Poisson distribution) that only one user passed the threshold. Hence, it is omitted for simplicity.

In this case, the rate of stream $i$ under MMSE receiver satisfies the following lower bound.
\begin{flalign*}
 \mathcal{R}_i (u_k) &= \text{E}\left[\left.\log\left(1+ P \bh_i^{\dagger}\bR \bh_i \right)\right| \Vert \bh_s \Vert^2 > u_k, \forall s \in \mathcal{S} \right]&
\\
&\stackrel{(a)}{=} \text{E}\left[\left.\log\left(1+ P \sum_{n=1}^r \lambda_n(\bR) \left\vert(U^{\dagger}\bh_i)_n\right\vert^2  \right)\right| \Vert \bh_s \Vert^2 > u_k, \forall s \in \mathcal{S} \right] &
\\
 &\stackrel{(b)}{=}  \text{E}\left[\left.\log\left(1+ P  \sum_{n=1}^r \frac{\left\vert(U^{\dagger}\bh_i)_n\right\vert^2}{\lambda_n\left(\bR^{-1}\right)}\right)\right| \Vert \bh_s \Vert^2 > u_k, \forall s \in \mathcal{S} \right] &
\\
&\stackrel{(c)}{=}  \text{E}\left[\left.\log\left(1+ P  \left(\sum_{n=1}^{r-j+1} \left\vert\left(U^{\dagger}\bh_i\right)_n\right\vert^2+ \sum_{n=r-j+2}^{r}\frac{\left\vert(U^{\dagger}\bh_i)_n\right\vert^2}{\lambda_n\left(\bR^{-1}\right)}\right)\right)\right| \Vert \bh_s \Vert^2 > u_k, \forall s \in \mathcal{S} \right] &
\\
 &\stackrel{(d)}{>}  \text{E}\left[\left.\log\left(1+ P \sum_{n=1}^{r-j+1} \left\vert\left(U^{\dagger}\bh_i\right)_n\right\vert^2 \right)\right| \Vert \bh_s \Vert^2 > u_k, \forall s \in \mathcal{S} \right] &
\\
 &\stackrel{(e)}{=}  \text{E}\left[\left.\log\left(1+ P \sum_{n=1}^{r-j+1} \left\vert\left(\bh_i\right)_n\right\vert^2 \right)\right| \Vert \bh_s \Vert^2 > u_k, \forall s \in \mathcal{S} \right] &
\\
 &=  \text{E}\left[\left.\log\left(1+ P \sum_{n=1}^{r-j+1} \left\vert\left(\bh_i\right)_n\right\vert^2 \right)\right| \Vert \bh_i \Vert^2 > u_k \right] &
\end{flalign*}
In the above chain of inequalities, (a) is since by \Cref{claim. R is unitarily invariant}, there is a unitary matrix $\bU$ such that $\bR = \bU \pmb{\Lambda} \bU^{\dagger}$ with $\pmb\Lambda = diag(\pmb\lambda_1,..,\pmb\lambda_r)$. As a result, we have the following quadric form:
$$\bh^{\dagger}\bR \bh = \bh^{\dagger}U \Lambda U^{\dagger} \bh = \sum_{n=1}^r \lambda_n(\bR)\left\vert\left(U^{\dagger}\bh\right)_n\right\vert^2 .$$
  (b) is since the eigenvalues of $\bR$ are the reciprocals of those of $\bR^{-1}$. (c) and (d) follow since, first, the dimensions of $\bH_{(-i)}\bH_{(-i)}^{\dagger}$ is $r \times r$ and its rank is $j-1$. Thus, $r-j+1$ of the eigenvalues of $\bR^{-1}$ (corresponding to the zero eigenvalues of $\bH_{(-i)}\bH_{(-i)}^{\dagger}$) are equal to one. Then, as $\bH_{(-i)}\bH_{(-i)}^{\dagger}$ is positive semidefinite, the non-zero eigenvalues are non negative. In fact, the eigenvalues of $\bH_{(-i)}\bH_{(-i)}^{\dagger}+I$ which are not unity, are relatively large as all columns of $\bH$ are above a threshold. (e) is by \Cref{claim. R is unitarily invariant} (ii).

Now, setting $\bx = \sum_{n=1}^{r-j+1} \left\vert\left(\bh_i\right)_n\right\vert^2$ and $\bx + \by = \Vert \bh_i \Vert^2$, \Cref{claim. lower bound on expectation} completes the proof.
\end{proof}

In Figure~\ref{fig: capacity scaling} we simulate the achievable sum-rate of the Channel-Access algorithm versus the number of users $K$, under the MMSE and the ZF receivers. The simulation results are compared to the upper and lower bounds. As can be seen, the MMSE receiver performance is indeed superior to the ZF receiver. Nevertheless, we see that both receivers attain the optimal scaling law. Interestingly, even though the bounds mostly apply to a large number of users, besides when the number of users is extremely small ($K=10$ users) all results fall within the bounds, and even for this small number of users case, the simulation results only slightly exceed the upper-bound and only for the ZF receiver.

In Figure~\ref{fig: mmse_sim} we simulate the  expected sum-rate of the {\scshape TB-Channel-Access} scheme, with MMSE-receiver and compare the simulation results to the upper and lower sum-rate bounds given in \Cref{lem. C mmse upper} and \Cref{lem. C mmse lower}, respectively.

In Figure~\ref{fig: ZFlower_vs_MMSElower} we compare  the MMSE receiver lower bound (dot-dashed line) and the ZF receiver lower bound (dashed line). As expected the MMSE achieves better performance when the threshold is lower, which corresponds to a lower SNR, and further, to a lower idle slot probability. On the other hand, as the threshold gets higher, which translates to high SNR, the ZF yields same results as MMSE. From this observation, since the MMSE receiver requires lower threshold to achieve higher sum-rate, we comprehend that the MMSE receiver is preferable to the ZF receiver, both in terms of sum-rate and idle slot probability, yet, still has a linear decoding complexity.

\begin{figure}
	\centering
		\includegraphics[scale=0.35,angle=-90]{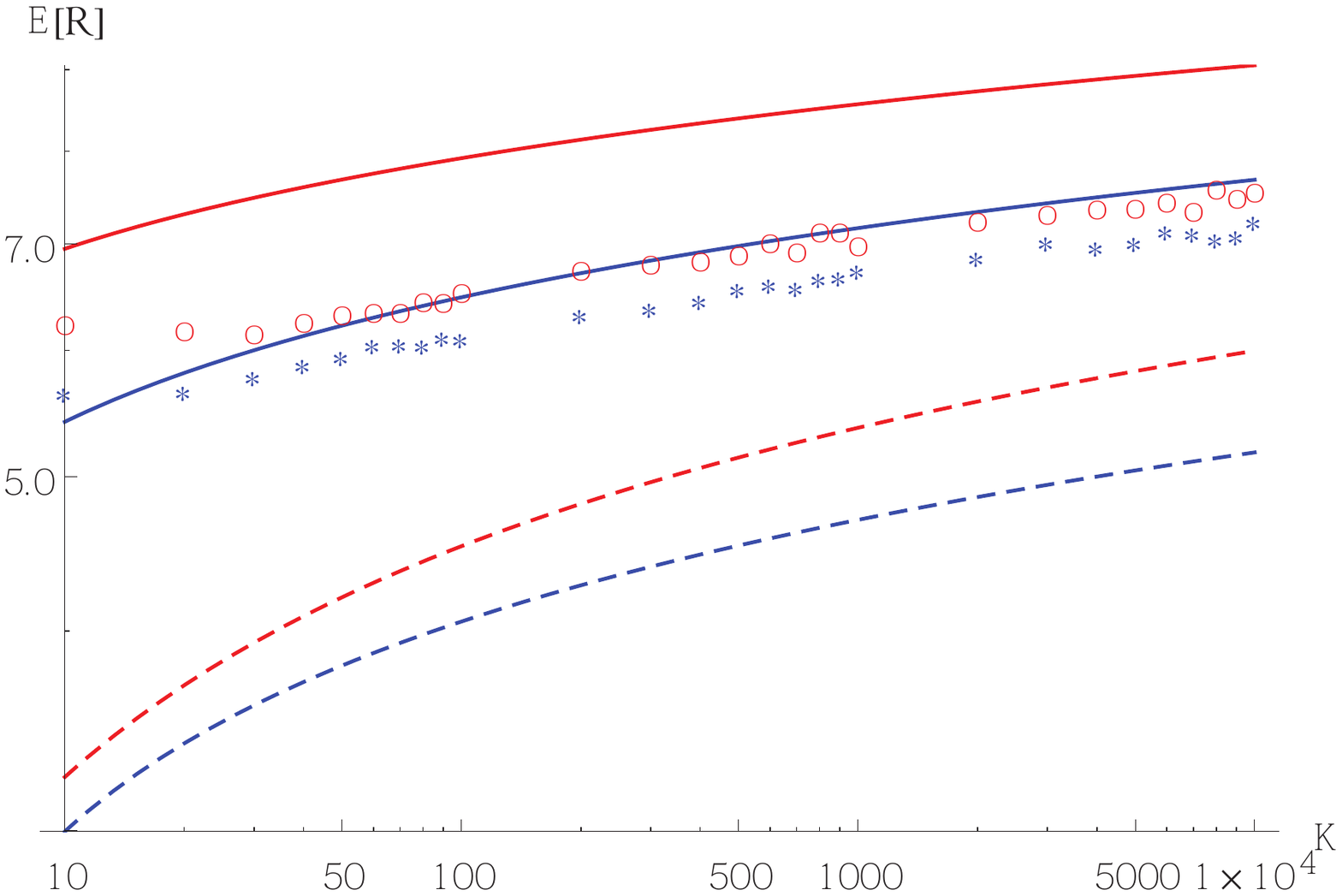}
	\caption{ZF and MMSE scaling-laws. The solid (dashed) red and blue lines represent the upper (lower) bound of the MMSE and ZF achievable rates, respectively, for different values of $K$, with $r=4$ antennas. The red circles and blue stars are based on simulation results, and depict the average sum-rate of the MMSE and ZF, respectively, upon $1500$ experiments. }
	\label{fig: capacity scaling}
\end{figure}

\begin{figure}
	\centering
		\includegraphics[scale=0.4,angle=-90]{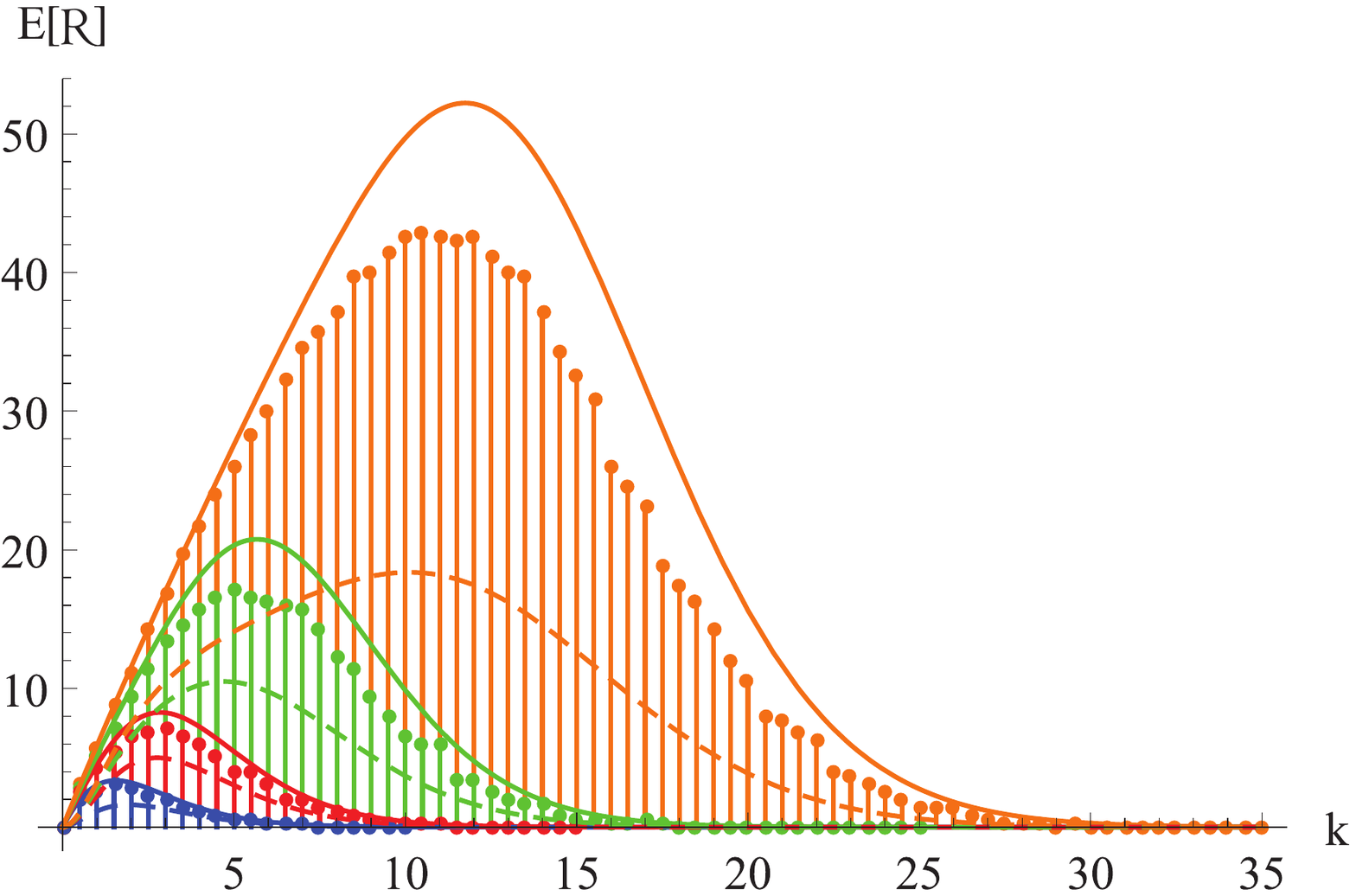}
	\caption{MMSE bounds. Bars are simulation results, while the solid and dashed lines represent upper and lower bound  results respectively. Expected sum-rate under a single threshold algorithm and $K=300$. The threshold $u_k$ is set such that $k$ users exceed it on average. Bars are the sum-rate under the {\scshape TB-Channel-Access} scheme simulation results under MMSE receiver, while the solid and dashed lines represent the upper and lower bounds, respectively. The orange (upper), green (2nd from above), red (3rd from above) and blue (lower) lobes are for $r=16,8,4$ and $2$ receive antennas, respectively. Note that the optimal $k$ is smaller than $r$. }
	\label{fig: mmse_sim}
\end{figure}

\begin{figure}%
	\centering
		\includegraphics[scale=0.35,angle=-90]{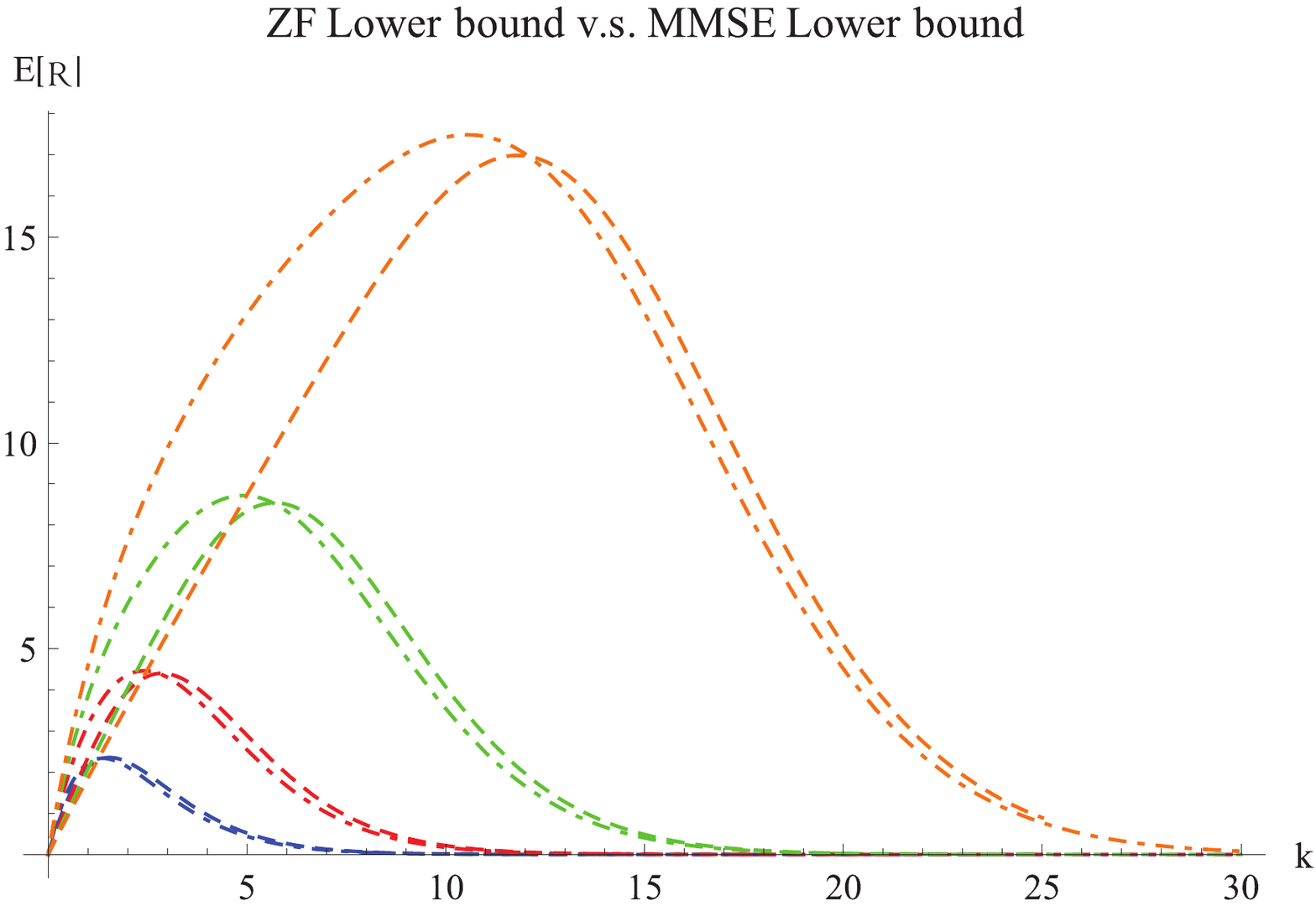}
	\caption{A comparison between the lower bound under an  MMSE receiver (dot-dashed) and the lower bound under the ZF (dashed) receiver. The orange (upper), green (2nd from above), red (3rd from above) and blue (lower) lobes are for $r=16,8,4$ and $2$ receive antennas, respectively. Note that since one should choose $k$ for best results, the MMSE lower bound predicts better results, achieved for a smaller $k$ per $r$ (the number of receive antennas). }
	\label{fig: ZFlower_vs_MMSElower}
\end{figure}
\section{Conclusion}\label{sec. conc}
In this paper, we analysed a distributed multiuser scheduling algorithm which utilizes the multiuser diversity and achieves the optimal \RED{capacity} \BLUE{sum-rate} scaling laws (for large number of users).
Specifically, we examined threshold-based algorithm, and characterized the scaling law of the expected sum-rate under linear detection, e.g., ZF and MMSE. To support the results, we provided both tractable analysis which gave insightful results as well as simulations which showed tightness even for moderate number of users.
We concluded that the distributed algorithm achieves the optimal scaling laws, i.e., $\Theta(r \log \log K)$.
\appendix
\section{Appendix}\label{sec:proofs}
The methods discussed in the previous sections are based on a common baseline procedure. First, given that the norm threshold was exceeded, we wish to get a handle on the \RED{capacity} \BLUE{rate} distribution. Then, we wish to express the rate at which users exceed the threshold, in order to understand how likely is that a given number of users exceed a certain threshold. Finally, a threshold is estimated. This threshold is set according to the fraction of users which are required to exceed it on average. In this paper, the threshold is set on the channel vector norm, either directly or after projecting it on the null-space of the already chosen users.
Note that we used a similar approach to analyse the expected rate, when the channel capacity is approximately Gaussian in~\cite{kampeas2012opportunistic}.
 In this section, we discuss the following problems directly for the problem at hand. 
\subsection{Derivation of the Normalizing Constants $a_{\{K,r\}}$ and  $b_{\{K,r\}}$}\label{proof:normalizing_constants}
The $\chi^2$ distribution is a special case of the gamma distribution. I.e., if $\bz \sim \chi^2_{2r}$ then $\bz \sim \Gamma(r ,\beta = 2)$, where $\Gamma(r ,\beta = 2)$ is the Gamma distribution with shape parameter $r$ and rate parameter $\beta$.
 According to the EVT, $b_K$ is the $1-1/K$ quantile, i.e., $1-F_{\chi^2}(b_n)=1/K$, and the corresponding $a_K$ is equals to $h(b_K)$, where $h(z)$ is the reciprocal hazard function \cite{de2006extreme,leadbetter1983}
\begin{equation*}\label{eqn: h(x)}
  h(z)=\frac{1-F(z)}{f(z)} \textmd{ for } z_F \leq z \leq z^F,
\end{equation*}
where \mbox{$z_F = \inf \{z:F(z)>0\}$ and $z^F = \sup \{z: F(z)< 1\}$} are the lower and upper endpoints of the ancestor distribution, respectively.
Accordingly, for the $b_{\{K,r\}}$ constant we consider the $1-1/K$ quantile of the  Gamma distribution, which can be obtained by using the inverse of the regularized upper incomplete gamma function. In particular, $b_{\{K,r\}} = \beta Q^{-1}\left(r,\frac{1}{K}\right)$ yields the $1-1/K$ quantile of the  Gamma distribution. To attain the $a_{\{K,r\}}$ constant, let us examine the hazard function $h(z)$ of the Gamma distribution.
\begin{eqnarray*}
h(z/ \beta) &=& \frac{1-F_{\Gamma}(z/\beta)}{f_{\Gamma}(z/\beta)}\\
     &=&  \beta \e^{z/\beta } (z/\beta)^{1-r}\Gamma(r)(1-F_{\Gamma}(z/\beta)).
\end{eqnarray*}
Accordingly,  for $z = b_{\{K,r\}}$ we obtain,
\begin{eqnarray*}
a_{\{K,r\}} &=& h\left(b_{\{K,r\}}/\beta\right)
\nonumber
\\
 &=& \frac{\beta}{K} \Gamma(r) \exp\left\{Q^{-1}\left(r,\frac{1}{K}\right)\right\} Q^{-1}\left(r,\frac{1}{K}\right)^{1-r}.
\end{eqnarray*}
\subsection{Threshold Arrival Rate and  Channel Gain Tail Distribution}\label{sec: tail dist}
Once a threshold is set, it is important to evaluate both the distribution of the number of users which exceed it as well as the conditional distribution of the norm given that the threshold was exceeded. Herein, we utilize the Point Process Approximation \cite{smith1989extreme} and its specific usage for threshold arrival rates in the \emph{single user case} \cite{kampeas2012opportunistic} in order to derive these distributions for the problem at hand.

Assume that $\bz_1,...,\bz_K$ is a sequence of i.i.d.\ random variables with a distribution function $F(z)$, such that $F(z)$ is in the domain of attraction of some GEV distribution G, with normalizing constants $a_K$ and $b_K$.
We construct a sequence of points $P_1,P_2,...$ on $[0,1] \times \R$ by
$$ P_K = \left\{ \left(\frac{i}{K},\frac{\bz_i - b_K}{a_K}\right) , i=1,2,...,K \right\},$$
and examine the limit process, as $K \rightarrow \infty$. Consider $P_K$ on the set $[0,1] \times (b_l + \epsilon , \infty)$, where \blue{$b_l$ is some floor value and} $\epsilon > 0$. According to Kallenberg's theorem \cite{kallenberg1983random,smith1989extreme}, $\lim_{K \to \infty} P_K \to P$, where $P$ is a non-homogeneous Poisson process with intensity density
$ \lambda(t,z) = (1 + \xi z)_{+}^{-\frac{1}{\xi} - 1}$, \blue{at sample value $z$,} \red{is the sample value,} and $t$ is the index of occurrence. In fact, in the i.i.d.\ case, the process intensity density $\lambda(t,z)$ is independent of $t$.

Let $\Lambda(B)$ be the expected number of points in the set $B$. $\Lambda(B)$  can be obtained by integrating the intensity of the Poisson process over $B$, That is $\Lambda(B) = \int_{b \in B}{\lambda(b)db}$. As we are interested in sets of the form $B_v = [0,1] \times (v,\infty)$, where $v>b_l$ (threshold exceedance), we have $\Lambda(B_v) = (1 + \xi v)_{+}^{-1/\xi}$, where $a_{+}$ denotes $\max \{0,a\}$ (e.g., \cite{kampeas2012opportunistic}). That is, the number of users whose channel norms exceed a threshold can be modeled by a Poisson process, with parameter $\Lambda\left(B_v\right)$.
Note that for $\xi \to 0$, as in the case of the $\chi^2$ distribution, we have $\Lambda\left(B_v\right) \to \e^{-v}$, hence setting a threshold $\log K -\log k$ results in a per-user arrival rate of $\frac{k}{K}$, and, as a result, a total arrival rate of $k$.

To compute the conditional distribution \blue{of the channel norm, given that it is above the threshold, at the limit of large $K$}, we proceed similar to ~\cite{kampeas2012opportunistic}. \blue{Specifically, focusing on points of the process $P_K$ that are above a}\red{For any} fixed \blue{threshold} $v > b_l$, let $u_K(v) = a_K v + b_K$, and let $z>0$. Then,
\begin{align}\label{eq. ppa convergence}
\blue{\lim_{K \to \infty}}&\Pr\left(\bz_i > a_K z+ u_K(v)| \bz_i  > u_K (v)\right)&\\
&\blue{= \lim_{K \to \infty} \Pr\left(P_K\prnt{\frac{i}{K}} > z+ v| P_K\prnt{\frac{i}{K}}> v\right)}&\nonumber\\
&= \Pr(P(t) > z + v|P(t) > v)
\nonumber\\
&= \frac{\Lambda(B_{z+v})}{\Lambda(B_v)}
\nonumber\\
&= \left[\left(1+\xi \frac{z}{1 + \xi v }\right)_{+}\right]^{-1/\xi},\nonumber
\end{align}
where  $P_K\prnt{\frac{i}{K}}$ \blue{denotes the process value at time $\frac{i}{K}$, that is,} \red{and $P(t)$ are the corresponding excess value} $P\prnt{\frac{i}{K},\frac{\bz_i - b_K}{a_K}}$ at index $\frac{i}{K}$ \blue{for any finite $K$, and $P(t)$ is the corresponding value for the limit process, with $t = \lim_{K\to \infty} \frac{i}{K}$}.  \red{Namely, the result is} \blue{The result is exactly} a generalized Pareto distribution, $GPD(\sigma_v,\xi)$, where $\sigma_v = 1 + \xi v$. \blue{In other words, at the limit of large $K$, the conditional distribution of the norm, given that it is above a high enough threshold, converges to a GPD. However, this is only convergence in distribution. In general, it does not directly result in \emph{convergence of the mean}. Yet, by the claim below, the distributions at hand converge monotonically to the limit, and, as a result, the limit of the expectations exists and equals the expectation at the limit.
\begin{claim}\label{claim_monotone} The above-threshold tail distributions in \eqref{eq. ppa convergence} converge monotonically (in $K$) to the GPD.
\end{claim}
\begin{proof}
The above-threshold tail distribution of the $\chi^2 (2r)$ distribution in \eqref{eq. ppa convergence}, for a given threshold $u_K(v)$, can be expressed using the upper incomplete gamma function as 
\begin{eqnarray*}
\Pr\prnt{\bz_i > a_K z + u_K(v) \big| \bz_i > u_K(v)} &=& \frac{Pr\prnt{\bz_i > a_K z + u_K(v)}}{Pr\prnt{\bz_i > u_K(v) }}\\
&=& \frac{\Gamma\prnt{r,a_K z + u_K(v)}}{\Gamma\prnt{r,u_K(v)}},
\end{eqnarray*}
where $a_K$ and $b_K$ are given in \eqref{eqn: a_n normalized slow} and \eqref{eqn: b_n normalized slow}, respectively. Note that $u_K(v)$ scales up with $K$ as $\Theta(\log K)$. To ease notation, hereafter we use $u_K$ to denote $u_K(v)$.

Consider a series of above-threshold tail distributions, where the threshold $u_K$ grows monotonically with $K$. Then, for monotone convergence it is enough to show that 
\begin{align*}
\frac{\partial }{\partial u_K}\frac{\Gamma(r,a_K z + u_K)}{\Gamma(r,u_K)} &= \frac{\frac{\partial }{\partial u_K}\Gamma(r,a_K z + u_K)\cdot\Gamma(r,u_K)-\frac{\partial }{\partial u_K}\Gamma(r,u_K)\cdot\Gamma(r,a_K z+ u_K)}{\Gamma(r,u_K)^2}
\\
&\leq 0.
\end{align*}
We focus on the enumerator, take the derivatives of the gamma functions, and conclude that one has to show that
$$(a_K z+u_K)^{r-1} \e^{-(a_K z + u_K)}\int_{u_K}^{\infty}t^{r-1} \e^{-t}dt-u_K^{r-1} \e^{-u_K}\int_{a_K z+ u_K}^{\infty} t^{r-1} \e^{-t}dt\geq 0.$$
Changing variables to $t = \tau+u_K$ and $t = \tau +a_K z +u_K$ in the left and right integrals, respectively, the above simplifies to
$$\int_{0}^{\infty}\left[(\tau + u_K)^{r-1} (a_K z + u_K)^{r-1}-u_K^{r-1}(\tau +a_K z + u_K)^{r-1}\right] \e^{-\tau}d\tau \geq 0$$
which is clearly true since $(\tau+u_K)(a_K z+ u_K)\geq u_K(\tau+a_K z+ u_K)$ when $a_K z\tau \geq 0$. This completes the proof.
\end{proof}
The convergence of the mean excess from the threshold to the mean of the limiting generalized Pareto distribution, can now follow from the monotone convergence theorem~\cite[Ch.\ 1]{durrett2010probability}. The monotone convergence is also clearly visible in \Cref{monotone_convergence}.
\begin{figure}
\centering
\includegraphics[scale=0.55]{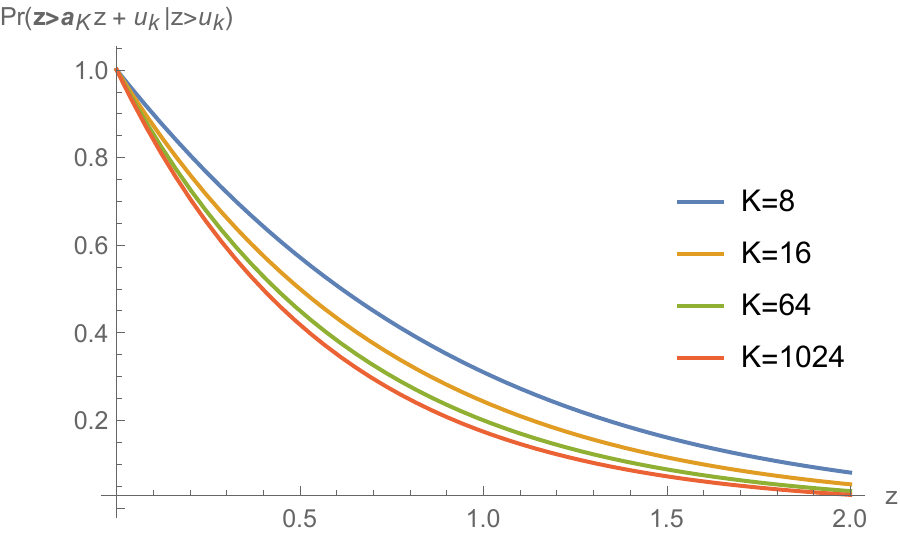}
\caption{Monotone convergence of the excess over the threshold distribution \eqref{eq. ppa convergence}.}
\label{monotone_convergence}
\end{figure}

Since the above-threshold tail is GPD, by \cite[Ch. 4.3.1]{coles2001introduction}, we now have $\E\brkt{\frac{\bz_i - u_k}{a_{K}} \big|\bz_i > u_k} \to \frac{\sigma_v}{1-\xi}$, and hence we approximate the expectation of the tail (for large enough $K$) as $\E\brkt{\bz_i \big|\bz_i > u_k} = u_k + \frac{a_K \sigma_v}{1-\xi}$.}
\red{Note that}\blue{Moreover, focusing on the excess values over the threshold, since the channel gain has a $\chi^2$ distribution, where $\xi \to 0$, by setting $\alpha = a_K z$ in \eqref{eq. ppa convergence},} \red{and}we obtain the exponential approximation, that is
\begin{equation*}\label{eqn: Exponential CDF}
    \Pr(\bz - u_K (v) \leq \alpha | \bz - u_K (v)>0) = 1 - \e^{-\frac{\alpha}{a_K}}
\end{equation*}
for all $\alpha \geq 0$ and large enough $K$.
That is, the \blue{above-threshold} tail distribution can be approximated using an exponential distribution with rate $\lambda = 1/a_K$. \blue{This results in the following corollary}.
 \red{As a result, taking the expected value, we obtain the following corollary.}
\begin{corollary}[\mbox{\cite[Ch. 4.3.1]{coles2001introduction}}]\label{coro: expected capacity above est by GEV thr}
Given a random complex Gaussian vector, and a \blue{sufficiently high} threshold such that \blue{the number of users that exceed it on average is $k$ out of $K$}, then
$$ \E\left[\Vert \bh \Vert^2 \bigg| \Vert \bh \Vert^2 > u_k  \right] = u_k + a_{\{K,r\}},$$
where $a_{\{K,r\}}$ is given in \eqref{eqn: a_n normalized}.
\end{corollary}
\subsection{Threshold Estimation}\label{sec: thr est}
Let $u_k$ be a threshold such that only the $k$ users with the strongest channel norm $\Arrowvert \bh_i \Arrowvert^2$ will exceed it on average.
Since $\Arrowvert \bh_i \Arrowvert^2$ follows $\chi^{2}_{2r}$ distribution, $u_k$ can be easily estimated using the Inverse-Gamma function. That is,
\begin{equation*}\label{eqn: p threshold igamma}
u_{k}  =  2 Q^{-1}(r,k/K).
\end{equation*}
Note, however, that the expression above does not give any intuition on the actual threshold value, or, more importantly, how it scales with $K$ for fixed $k$.
Yet, note that the threshold $u_k$ is closely related to the normalizing constant $b_{\{K,r\}}$, since both aim to capture the last quantiles of the ancestor distribution.
The exact relation can be obtain using the Point process approximation. Specifically, the \RED{exceedance} \BLUE{exceeding} rate for the $\chi^2$-distribution is
$$\Lambda(u_k) = \e^{-\frac{u_k- b_{\{K,r\}}}{a_{\{K,r\}}}}.$$
Accordingly, choosing $u_k = - a_{\{K,r\}}\log k + b_{\{K,r\}}$ gives a rate of $k$ in total. Thus, letting $a_{\{K,r\}}$ and $b_{\{K,r\}}$ to be set according
to \eqref{eqn: a_n normalized} and \eqref{eqn: b_n normalized}, respectively, $u_k$ is set. Now, to see the scaling law, note that $a_{\{K,r\}}$  and $b_{\{K,r\}}$ in \eqref{eqn: a_n normalized} and \eqref{eqn: b_n normalized} converges to $a_K$ and $b_K$ in \eqref{eqn: a_n normalized slow} and (\ref{eqn: b_n normalized slow}), respectively, hence $u_k = O(\log K)$. In \Cref{fig:an_convergence} Figure \ref{fig:convergence of b} depicts the above convergence result. 

\begin{figure}
	\centering
		\includegraphics[scale=0.55]{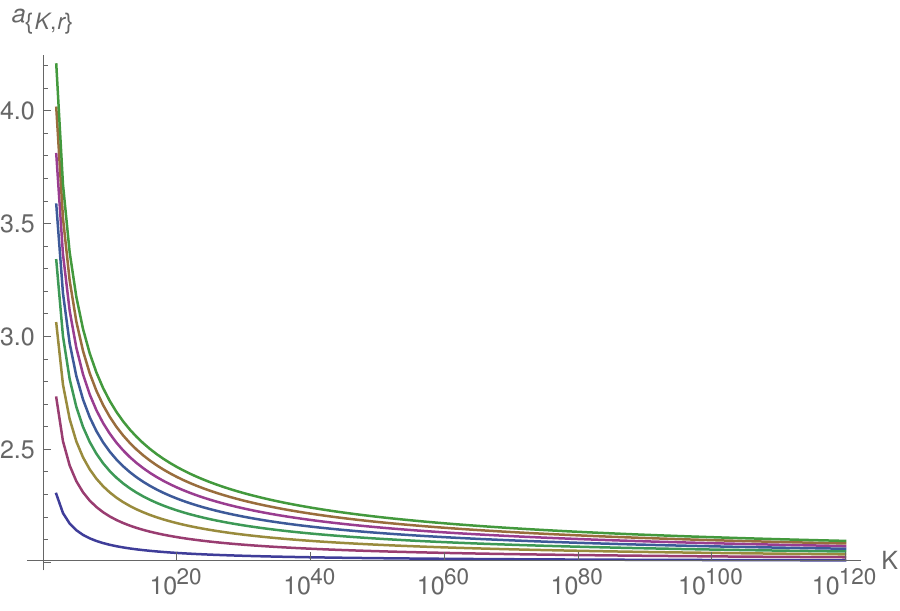}
	\caption{Convergence of the normalizing constant $a_{\{K,r\}}$ for $\{2r\}_{r=1}^8$ receiving antennas to the limit $a_K =2$. Note that this slow convergence does not affect the convergence of the actual distribution, as the values of $a_{\{K,r\}}$ are within a constant factor from $2$ even for very moderate $K$.}
	\label{fig:an_convergence}
\end{figure}

\begin{figure}
\centering
\includegraphics[scale=0.55]{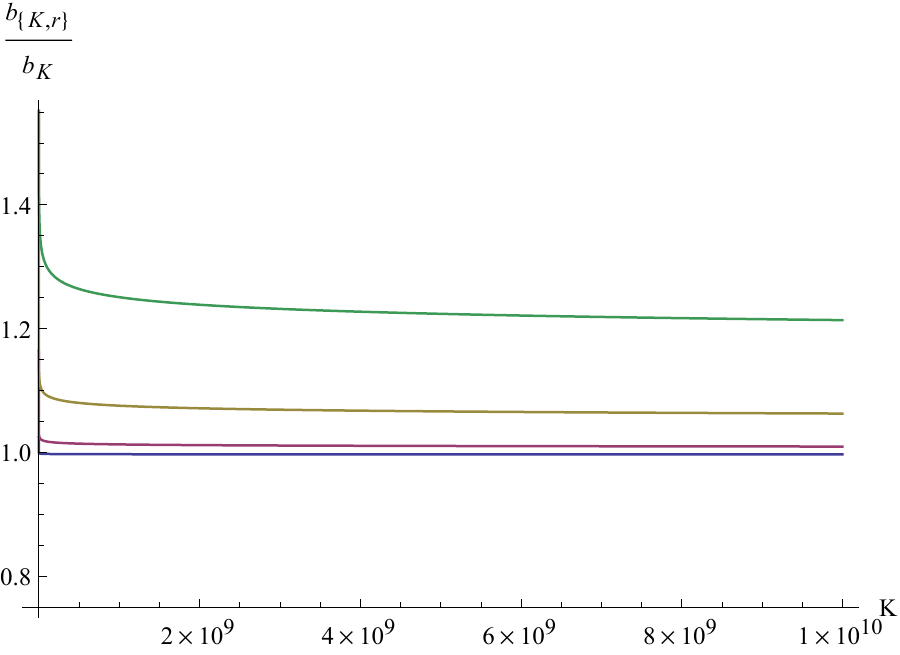}
\caption{The ratio between the normalizing constant $b_{\{K,r\}}$ in \eqref{eqn: b_n normalized} and the value in (\ref{eqn: b_n normalized slow}). The curves are plotted for $r=16,8,4$ and $2$, top to bottom.}
\label{fig:convergence of b}
\end{figure}


\subsection{Proof of \Cref{C ZF uniform users lower}}\label{proof:ZF_lower}
Following the derivations of the upper bound, we have
\begin{align*}
\E\brkt{\mathcal{R}(u_k)}
&= \sum_{j=1}^{r}\frac{k^{j}\e^{-k}}{j!} \sum_{i=1}^{j} \text{E}\left[\log\left(1 + P \left. \sum_{m=1}^{r-j+1} \Arrowvert \bh_i \Arrowvert^2 \frac{\vert\langle\bV_i^{(m)},\bh_i\rangle \vert^2}{\Arrowvert \bh_i \Arrowvert^2 \Arrowvert \bV^{(m)}_i \Arrowvert^2} \right)\right| \Vert \bh_i \Vert^2 > u_k \right] \blue{+ O\prnt{\frac{\log \log K}{K}}}
\\
&\stackrel{(a)}{\ge} \sum_{j=1}^{r}\frac{k^{j}\e^{-k}}{j!} \sum_{i=1}^{j} \text{E}\log\left(1 + P u_k \sum_{m=1}^{r-j+1} \frac{\vert\langle\bV_i^{(m)},\bh_i\rangle\vert^2}{\Arrowvert \bh_i \Arrowvert^2 \Arrowvert \bV^{(m)}_i \Arrowvert^2} \right) \blue{+ O\prnt{\frac{\log \log K}{K}}}
\\
&\stackrel{(b)}{=} \sum_{j=1}^{r}\frac{k^{j}\e^{-k}}{j!}j \text{E}\log\left(1 + P u_k \sum_{m=1}^{r-j+1} \frac{\vert\langle\bV_{i'}^{(m)},\bh_{i'}\rangle\vert^2}{\Arrowvert \bh_{i'} \Arrowvert^2 \Arrowvert \bV^{(m)}_{i'} \Arrowvert^2} \right) \blue{+ O\prnt{\frac{\log \log K}{K}}}
\\
&\ge \sum_{j=1}^{r}\frac{k^{j}\e^{-k}}{j!} j \text{E}\log\left(1 + P u_k  \frac{\vert\langle\bV_{i'}^{(1)},\bh_{i'}\rangle\vert^2}{\Arrowvert \bh_{i'} \Arrowvert^2 \Arrowvert \bV^{(1)}_{i'}\Arrowvert^2} \right) \blue{+ O\prnt{\frac{\log \log K}{K}}}
\\
&\stackrel{(c)}{=} \sum_{j=1}^{r}\frac{k^{j}\e^{-k}}{j!} j \int_{0}^{1}(r-1)(1-\alpha)^{r-2}\log\left(1+P u_k \alpha\right) d \alpha \blue{+ O\prnt{\frac{\log \log K}{K}}}
\end{align*}
where (a) is since the norms of \emph{all users participating} are above the threshold $u_k$; (b) is since the angles in the inner sum are identically distributed and independent of $i$, hence an arbitrary $1 \leq i' \leq j$ can be used; (c) is by explicitly computing the expectation over the angle between $\bh_i$ and $\bV^{(1)}_i$, remembering that it has a density $(r-1)(1-\alpha)^{r-2}$ for $0 \leq \alpha \leq 1$. \blue{Omitting the $O\prnt{\frac{\log \log K}{K}}$}, completes the proof.
\blue{
\subsection{A Proof that the Sum Over the Poisson Distribution is Monotonically Increasing in $r$}\label{sec. poisson monoticity}
To prove that the sum over the Poisson distribution with parameter $r$, up to $r-1$, is monotonically increasing in $r$, one must show that 
$$ \sum_{j=0}^{r-1}\frac{r^j}{j!}\mathrm{e}^{-r} < \sum_{j=0}^{r}\frac{(r+1)^j}{j!}\mathrm{e}^{-(r+1)}, \forall r \in \mathbb{N}.$$
This can be rewritten as  
$$\frac{\Gamma(r,r)}{\Gamma(r)} < \frac{\Gamma(r+1,r+1)}{\Gamma(r+1)},$$
where $\Gamma(r,r) = \int_r^\infty t^{r-1} \mathrm{e}^{-t} \mathrm{d} t$ is the upper incomplete gamma function.
Equivalently, we need to show that
$$r \int_{r}^\infty t^{r-1} \mathrm{e}^{-t} \mathrm{d} t <  \int_{r+1}^\infty t^{r} \mathrm{e}^{-t} \mathrm{d} t.$$
Since the derivative of $t^{r} \mathrm{e}^{-t}$ is negative on the interval $[r,\infty)$, the integrand is decreasing on that interval and we have, 
$$\int_{r}^{r+1} t^{r} \mathrm{e}^{-t} \mathrm{d}t < r^{r} \mathrm{e}^{-r}, \qquad \forall r>0.$$ 
Accordingly,
\begin{eqnarray*}
\int_{r+1}^{\infty} t^{r} \mathrm{e}^{-t}  \mathrm{d}t &=&  \int_{r}^{\infty} t^{r} \mathrm{e}^{-t} \mathrm{d}t - \int_{r}^{r+1} t^{r} \mathrm{e}^{-t} \mathrm{d}t \\
&>&  \int_{r}^{\infty} t^{r} \mathrm{e}^{-t} \mathrm{d}t  - r^{r} \mathrm{e}^{-r}\\
&=&  r \int_{r}^{\infty} t^{r-1} \mathrm{e}^{-t} \mathrm{d}t,
\end{eqnarray*}
where the last step is due to integration by parts. This completes the proof.
}
\subsection{Proof of \Cref{claim. integral result}}\label{proof:integral_result}
We wish to solve the following definite integral:
\begin{equation*}
(r-1) \int_0^1 (1-\alpha)^{r-2} \log\left(1 + P u \alpha\right) d\alpha.
\end{equation*}
We first make the substitution $x = (1-\alpha)$. As a result, we obtain the following integral:
$$(r-1) \int_0^1 x^{r-2} \log\left(1 + P u (1- x)\right) dx.$$
Now, using integration by parts, we have
\begin{equation*}
(r-1) \int_0^1 x^{r-2} \log\left(1 + P u (1- x)\right) dx
= x^{r-1}\log \left(1+ P u(1-x)\right) \Bigg|_0^1 + \int_0^1 \frac{P u x^{r-1}}{1 + P u(1-x)}dx.
\end{equation*}
First, note that
\begin{equation*}
x^{r-1}\log \left(1+ P u(1-x)\right) \Bigg|_0^1=0.
\end{equation*}
Now, performing a polynomial long division, the term inside the integral can be expressed as
$$\frac{P u x^{r-1}}{1+ P u(1-x)} = -\sum_{i=0}^{r-2}\left(\frac{1+ P u}{P u}\right)^i x^{r-2-i} + \frac{(1+P u)^{r-1}}{(P u)^{r-2}(1+P u (1-x))}.$$
Finally, we exchange the integration with the (finite) summation, then integrate each term in the sum according to $dx$. We have:
\begin{align*}
\int_0^1 & \frac{P u x^{r-1}}{1 + P u(1-x)}dx
\\
&=-\sum_{i=0}^{r-2}\left(\frac{1+ P u}{P u}\right)^i \frac{1}{r-1-i}x^{r-1-i} \Bigg|_0^1 - \left(\frac{1+ P u}{P u}\right)^{r-1}\log\left(1 + P u (1-x)\right) \Bigg|_0^1.
\\
&=
\left(\frac{1 + P u}{P u}\right)^{r-1}\log\left(1+P u\right)-\sum_{i=0}^{r-2} \left(\frac{1+ P u}{P u}\right)^i \frac{1}{r-1-i}.
\end{align*}
Thus, \Cref{claim. integral result} follows.
\subsection{Proof of \Cref{claim. entries var of above thr vec}}\label{proof:above_thr_vec_entries}
 It can be shown by the total law of expectation, that the entries of the channel vector remain zero, given that the channel norm is greater than a threshold,  as follows.
\begin{eqnarray*}
\text{E}\left[ h_{i,n}  \Big| \Vert \bh_i \Vert^2 > u_k\right] &=& \text{E}\left[ h_{i,n}  \Big| \sum_{n=1}^r \vert h_{i,n} \vert^2 > u_k\right]\\
 &=& \text{E}\left[\left. \text{E}\left[ h_{i,n}  \Big| \sum_{n=1}^r \vert h_{i,n} \vert^2 > u_k, \sum_{\substack{m=1, \\ m\neq n}}^r \vert h_{i,m} \vert^2 = \tilde{u} \right]\right| \sum_{n=1}^r \vert h_{i,n} \vert^2 > u_k   \right]\\
 &=& \text{E}\left[\left. \underbrace{\text{E}\left[ h_{i,n}  \Big|  \vert h_{i,n} \vert^2 > u_k - \tilde{u} \right]}_{(\ast)}\right| \sum_{n=1}^r \vert h_{i,n} \vert^2 > u_k   \right]
\end{eqnarray*}
Since $h_{i,n}$ is complex Gaussian random variable, the inner expectation $(\ast)$ can be decomposed to a real and an imaginary part as
 \begin{eqnarray*}
\text{E}\left[ h_{i,n}  \big|  \vert h_{i,n} \vert^2 > u_k - \tilde{u} \right]&=& \text{E}\left[ a +i b  \big|  a^2+b^2 > u_k - \tilde{u} \right]\\
&=& \text{E}\left[ a \big|  a^2+b^2 > u_k - \tilde{u} \right]+ i \text{E}\left[ b  \big|  a^2+b^2 > u_k - \tilde{u} \right]\\
\end{eqnarray*}
To ease notation, let us denote $\tilde{\tilde{u}} \doteq \sqrt{u_k -\tilde{u} - b^2}$. Focusing on the real part, we use the total expectation law again to obtain:
\begin{eqnarray*}
\text{E}\left[ a \Big|  a^2+b^2 > u_k - \tilde{u} \right] &=& \text{E}\left[\left. \text{E}\left[a \Big| a^2 > u_k -\tilde{u} - b^2, b^2\right] \right| a^2 +b^2 > u_k - \tilde{u} \right]\\
&=& \text{E}\left[\left. \text{E}\left[a \Big| \vert a \vert > \tilde{\tilde{u}}, b^2\right] \right| a^2 +b^2 > u_k - \tilde{u} \right]\\
&=& \text{E}\left[\left.\underbrace{\text{E}\left[sign(a) \vert a \vert \Big| \vert a \vert > \tilde{\tilde{u}}, b^2\right] }_{(\ast \ast)} \right| a^2 +b^2 > u_k - \tilde{u} \right].
\end{eqnarray*}
Now, we show that the sign of $a$ is independent of the amplitude $\vert a \vert$ and $b^2$.
\begin{eqnarray*}
\Pr\left( sign(a) = -1 \Big| \vert a \vert > \tilde{\tilde{u}} , b^2 \right) &=& \Pr\left( a < 0 \Big| \vert a \vert > \tilde{\tilde{u}} ,b^2\right)\\
&=&\frac{\Pr\left( a < 0 , \vert a \vert > \tilde{\tilde{u}} , b^2 \right)}{\Pr\left( \vert a \vert > \tilde{\tilde{u}} , b^2 \right)}\\
&=& \frac{\Pr \left( a < - \tilde{\tilde{u}} , b^2 \right)}{\Pr \left( a < - \tilde{\tilde{u}} , b^2\right) + \Pr \left( a > \tilde{\tilde{u}} , b^2 \right)}\\
&=& \frac{1}{2}.
\end{eqnarray*}
Thus, the sign of $a$, the amplitude $\vert a \vert$ and $b^2$ are independent.
Consequently, the inner expectation $(\ast \ast)$ is equal to:
\begin{eqnarray*}
\text{E}\left[sign(a) \vert a \vert \big| \vert a \vert > \tilde{\tilde{u}}, b^2\right]&=& \text{E}\left[sign(a)\right]\text{E}\left[ \vert a \vert \big| \vert a \vert > \tilde{\tilde{u}}, b^2\right]\\
&=& 0.
\end{eqnarray*}
From symmetry, same result can be obtained for the imaginary part. Thus, the vectors are still zero mean vectors.

However, the variance of the vector entries increase, given that the vector norm is greater than a threshold. Remember, by EVT \cite{leadbetter1983}, the tail is exponentially distributed with rate parameter $1/a_{\{K,r\}}$. Thus, it follows that
\begin{eqnarray*}
u_k + a_{\{K,r\}}  &\stackrel{(a)}{=}& \text{E}\left[\left. \Vert \bh_i \Vert^2\right| \Vert \bh_i \Vert^2 > u_k\right] \\
&\stackrel{(b)}{=}& \sum_{n=1}^r \text{E}\left[\left. \vert h_{i,n} \vert^2\right| \Vert \bh_i \Vert^2 > u_k \right] \\
&\stackrel{(c)}{=}& r \text{E}\left[\left. \vert h_{i,n} \vert^2\right| \Vert \bh_i \Vert^2 > u_k \right]
\label{eq: element conditional expectation}
\end{eqnarray*} where (a) is the result of computing the expected norm of an i.i.d.\ complex normal random vector, given that it exceeded a threshold $u_k$. (b) follows from the linearity of expectation operator.  (c) is since the elements of $\bh_i$ are identically distributed. We point out that choosing vectors with norm greater than a threshold increase the variance of the entries to $ (u_k + a_{\{K,r\}})/r $.

Similar to the entries conditional expectation,  it can be shown that the matrix entries are still uncorrelated in pairs. Namely,
\begin{align*}
\text{E}\left[ h_{i,n}^{\ast}h_{i,m}  \big| \Vert \bh_i \Vert^2 > u_k\right]&\\
 &= \text{E}\left[ h_{i,n}^{\ast}h_{i,m}  \big| \sum_{n=1}^r \vert h_{i,n} \vert^2 > u_k\right]&\\
 &= \text{E}\left[\left. \text{E}\left[ h_{i,n}^{\ast}h_{i,m}  \big| \sum_{n=1}^r \vert h_{i,n} \vert^2 > u_k, \sum_{\substack{m=1 \\m\neq n}}^r \vert h_{i,m} \vert^2 = \tilde{u} \right]\right| \sum_{n=1}^r \vert h_{i,n} \vert^2 > u_k   \right]&\\
 &= \text{E}\left[\left. \underbrace{\text{E}\left[ h_{i,n}^{\ast}h_{i,m}  \big|  \vert h_{i,n} \vert^2 > u_k - \tilde{u} \right]}_{(\tilde{\ast})}\right| \sum_{n=1}^r \vert h_{i,n} \vert^2 > u_k   \right]&
\end{align*}
Similar to the method above, we decompose $h^{\ast}_{i,n}$ to its real and imaginary part. Thus, the inner expectation $(\tilde{\ast})$ is
\begin{align*}
 \text{E}\left[ h_{i,n}^{\ast}h_{i,m}  \Big|  \vert h_{i,n} \vert^2 > u_k - \tilde{u} \right]&\\
  &= \text{E}\left[ a h_{i,m} - i b h_{i,m}  \Big|  a^2 + b^2 > u_k - \tilde{u} \right]&\\
  &= \text{E}\left[ a h_{i,m} \Big|  a^2 + b^2 > u_k - \tilde{u} \right]- i \text{E}\left[ b h_{i,m} \Big|   a^2 + b^2 > u_k - \tilde{u} \right]&
\end{align*}
Using the total expectation law on the real part of the expectation, we have
\begin{align*}
\text{E}\left[ a h_{i,m} \Big|  a^2 + b^2 > u_k - \tilde{u} \right]&\\
&= \text{E}\left[\left. \text{E}\left[ a h_{i,m} \Big|  a^2  > u_k - \tilde{u} - b^2, b^2 \right] \right|   a^2 + b^2 > u_k - \tilde{u}  \right]&\\
&= \text{E}\left[\left. \underbrace{\text{E}\left[ sign(a) \vert a \vert h_{i,m} \Big|  \vert a \vert  > \tilde{\tilde{u}}, b^2 \right]}_{(\tilde{\ast \ast})} \right|   a^2 + b^2 > u_k - \tilde{u}  \right].&
\end{align*}
Now, similar to the above method,  let us show that the sign of $a$ is independent of  $\vert a \vert$ and $h_{i,m}$.
\begin{eqnarray*}
\Pr\left( sign(a) = -1 \Big| \vert a \vert > \tilde{\tilde{u}} , b^2,  h_{i,m} \right) &=& \Pr\left( a < 0 \Big| \vert a \vert > \tilde{\tilde{u}} ,b^2,  h_{i,m}\right)\\
&=&\frac{\Pr\left( a < 0 , \vert a \vert > \tilde{\tilde{u}} , b^2, h_{i,m} \right)}{\Pr\left( \vert a \vert > \tilde{\tilde{u}} , b^2, h_{i,m} \right)}\\
&=& \frac{\Pr \left( a < - \tilde{\tilde{u}} , b^2, h_{i,m} \right)}{\Pr \left( a < - \tilde{\tilde{u}} , b^2, h_{i,m}\right) + \Pr \left( a > \tilde{\tilde{u}} , b^2 , h_{i,m}\right)}\\
&=& \frac{1}{2} = \Pr(a < 0).
\end{eqnarray*}
Thus, the inner expectation $(\tilde{\ast \ast})$ can be expressed as
\begin{eqnarray*}
\text{E}\left[sign(a) \vert a \vert h_{i,m} \big| \vert a \vert > \tilde{\tilde{u}}, b^2\right]&=& \text{E}\left[sign(a)\right]\text{E}\left[ \vert a \vert h_{i,m} \big| \vert a \vert > \tilde{\tilde{u}}, b^2\right]\\
&=& 0.
\end{eqnarray*}
From symmetry, same result can be obtained for the imaginary part. Thus, given that the norm is greater than a threshold, the vector entries sill are uncorrelated.
 Thus, Claim~\ref{claim. entries var of above thr vec} follows.
\subsection{Proof of \Cref{claim. R is unitarily invariant}}\label{proof:R_unitarily_invariant}
A Hermitian random matrix is unitary invariant if the joint distribution of its entries does not change under unitary transformation, namely, $\bH \sim U \bH$ for any unitary matrix $U$. It is well known that a Gaussian matrix is unitary invariant (e.g., \cite{telatar1999capacity}).
It is not hard to see that $\bH \bH^{\dagger} \sim U \bH  \bH^{\dagger} U^{\dagger}$. That is, the Wishart matrix is also unitary invariant.
Similarly, as $U I U^\dagger = I$, $\left(\bH \bH^{\dagger} + I \right) \sim U \left(\bH \bH^{\dagger} + I \right)U^{\dagger}$.

Since $\left(\bH \bH^{\dagger} + I \right)$ is unitary invariant,
by \cite[Lemma 2.6]{tulino2004random}, there is a decomposition such that $\left(\bH \bH^{\dagger} + I \right)= \bU \pmb{\tilde{\Lambda}}\bU^{\dagger}$, $\bU$ being unitary and independent of $\pmb{\tilde{\Lambda}}$, hence $\left(\bH \bH^{\dagger} + I \right)^{-1}= \bU \pmb{\tilde{\Lambda}}^{-1} \bU^{\dagger}$ and part (i) follows with $\pmb{\Lambda} = \pmb{\tilde{\Lambda}}^{-1}$. Note that if we are only interested in the decomposition, and the independence between the unitary matrix and the eigenvalues is not required (as is the case in the lower bound below), the proof is simpler since $\bR$ is Hermitian.

For part (ii), note that if $\bh \sim U\bh$, then $\bh |\hspace{0.15cm} \Vert \bh_i \Vert^2 > u \sim U\bh |\hspace{0.15cm} \Vert \bh_i \Vert^2 > u$. This is since the unitary transform does not change the vector norm, hence the set of instances $h \in \C^r$ with $\Vert h \Vert^2 >u$ is the same as those with $\Vert U h \Vert^2 >u$. The rest follows in the same manner.
\subsection{Proof of Claim~\ref{claim. lower bound on upper gamma s geq 2}}
First, note that,
$$\Gamma(s,x) = \int_{x}^{\infty} \tau^{s-1}\e^{-\tau} d\tau = \e^{-x} \int_0^{\infty} (t+x)^{s-1}\e^{-t} dt = \e^{-x} \text{E}[(T+x)^{s-1}],$$
where $T \sim \exp(1)$.

Accordingly, since $(t+x)^{s-1}$ is convex in $t$ for the corresponding $s$, we apply Jensen's inequality to the r.h.s.\ and note that $\text{E}[T] =1$, proving the bound.
\subsection{Proof of \Cref{claim. lower bound on expectation}}
We wish to bound the following expectation, when $\bx \sim \chi^2_{2(r-j+1)}$ and $\by \sim \chi^2_{2(j-1)}$, independent of $\bx$:
\begin{multline*}
    \E[\log(1+\bx)|\bx +\by > u]
     = \frac{1}{\Pr\left(\bx +\by >u\right)}
\\
\left[\int_{x=0}^{u} \int_{y=u-x}^{\infty}\log(1+x)f_{\bx}(x) f_{\by}(y) dy dx + \int_{x=u}^{\infty} \int_{y=0}^{\infty}\log(1+x)f_{\bx}(x) f_{\by}(y)dy dx \right].
\end{multline*}
Performing integration on $y$, we have the following.
\begin{flalign}\label{eq. mmse lower bound}
&\E[\log(1+x)|\bx +\by > u] &\nonumber\\
  & \quad \quad = \frac{1}{\Pr\left(\bx +\by >u\right)}\left[ \underbrace{\int_{x=0}^{u} \frac{\Gamma(j-1,(u-x)/2)}{\Gamma(j-1)} \log(1+x)f_{\bx}(x) dx}_{I_2(x)} + \underbrace{\int_{x=u}^{\infty}\log(1+x)f_{\bx}(x) dx}_{I_1(x)} \right]&
\end{flalign}
where $\Gamma(s,x) = \int_{x}^{\infty} t^{s-1}\e^{-t} dt$ is the upper incomplete Gamma function.
Further, note that the complement CDF of the $\chi^2_{2(j-1)}$-distribution is $\Pr\left(\by > u-x\right) = \frac{\Gamma(j-1,(u-x)/2)}{\Gamma(j-1)}$.

Accordingly, to evaluate (\ref{eq. mmse lower bound}), we need to evaluate both $I_1(x)$ and $I_2(x)$. We begin with integration by parts on $I_1(x)$, where the anti-derivative is $f_{\bx}(x)$. We have
\begin{eqnarray*}
I_1(x) &=& \log(1+x)\int f_{\bx}(x)dx \Bigg|_{x=u}^{\infty} - \int_{x=u}^{\infty}\frac{\int f_{\bx}(x) dx}{1+x}dx\\
&=& \log(1+x)\left(-\frac{\Gamma(r-j+1,x/2)}{\Gamma(r-j+1)}\right)\Bigg|_{x=u}^{\infty} + \frac{1}{\Gamma(r-j+1)}\int_{x=u}^{\infty} \frac{\Gamma(r-j+1,x/2)}{1+x} dx,
\end{eqnarray*}
where the last equality follows since $\int f_{\bx}(x) dx = \int \frac{x^{r-j} \e^{-x/2}}{2^{r-j+1} \Gamma(r-j+1)} dx = -\frac{\Gamma(r-j+1,x/2)}{\Gamma(r-j+1)}$.
Now, to bound the above expression, we restrict our attention to $1<j\leq r$. Since $\lim_{x\to\infty}\log(1+x)\Gamma(s,x/2)=0$, utilizing Claim~\ref{claim. lower bound on upper gamma s geq 2}, we have
\begin{eqnarray*}
I_1(x)&\geq& \frac{\Gamma(r-j+1,u/2)}{\Gamma(r-j+1)} \log(1+u) +\frac{1}{\Gamma(r-j+1)}\int_{x=u}^{\infty} \frac{\e^{-x/2} (1+x/2)^{r-j}}{1+x} dx \\
&>& \frac{\Gamma(r-j+1,u/2)}{\Gamma(r-j+1)} \log(1+u) +\frac{1}{\Gamma(r-j+1)}\int_{x=u}^{\infty} \frac{\e^{-x/2} (1+x/2)^{r-j}}{2+x} dx \\
&=& \frac{\Gamma(r-j+1,u/2)}{\Gamma(r-j+1)} \log(1+u) +\frac{1}{2\Gamma(r-j+1)}\int_{x=u}^{\infty} \e^{-x/2} (1+x/2)^{r-j-1} dx
\\
&=& \frac{\Gamma(r-j+1,u/2)}{\Gamma(r-j+1)} \log(1+u) +\frac{e}{\Gamma(r-j+1)}\int_{z=1+u/2}^{\infty} \e^{-z} z^{r-j-1} dz \\
&=& \frac{\Gamma(r-j+1,u/2)}{\Gamma(r-j+1)} \log(1+u) +\frac{e}{\Gamma(r-j+1)} \Gamma(r-j,1+u/2).
\end{eqnarray*}

Next, let us address $I_2(x)$. Again, using Claim~\ref{claim. lower bound on upper gamma s geq 2}, we have
\begin{eqnarray*}
I_2(x) &\geq&  \frac{1}{\Gamma(j-1)}\int_{x=0}^{u} \e^{-\frac{u-x}{2}}\left(1+ \frac{u-x}{2}\right)^{j-2} \log(1+x)f_{\bx}(x) dx \\
 &=&   \frac{1}{\Gamma(j-1)} \int_{x=0}^{u} \e^{-\frac{u-x}{2}}\left(1+ \frac{u-x}{2}\right)^{j-2} \log(1+x) \frac{x^{r-j} \e^{-x/2}}{2^{r-j+1} \Gamma(r-j+1)} dx \\
 &=&  \frac{\e^{-u/2}}{2^{r-j+1} \Gamma(r-j+1) \Gamma(j-1)} \int_{x=0}^{u} \left(1+ \frac{u-x}{2}\right)^{j-2}  x^{r-j} \log(1+x)  dx \\
 &=&  \frac{\e^{-u/2}}{2^{j-2} 2^{r-j+1} \Gamma(r-j+1) \Gamma(j-1)} \int_{x=0}^{u} \left(2+ u-x\right)^{j-2}  x^{r-j} \log(1+x)  dx \\
 &>&  \frac{\e^{-u/2}}{ 2^{r-1} \Gamma(r-j+1)\Gamma(j-1)} \int_{x=0}^{u} \left(u-x\right)^{j-2}  x^{r-j} \log(1+x)  dx \\
 &=&  \frac{\e^{-u/2} u^{r-2}}{2^{r-1} \Gamma(r-j+1)\Gamma(j-1)} \int_{x=0}^{u} \left(1-\frac{x}{u}\right)^{j-2}  \left(\frac{x}{u}\right)^{r-j} \log(1+x)  dx \\
\end{eqnarray*}
Let $z = x/u$. Accordingly, we have
\begin{eqnarray*}
I_2(x)
&> &\frac{\e^{-u/2} u^{r-1}}{ 2^{r-1} \Gamma(r-j+1) \Gamma(j-1)} \int_{z=0}^{1} \left(1-z\right)^{j-2}  z^{r-j} \log(1+u z)  dz\\
&\geq& \frac{\e^{-u/2} u^{r-1}\beta(r-j+1,j-1)}{ 2^{r-1} \Gamma(r-j+1)\Gamma(j-1)} \int_{z=0}^{1} \frac{1}{\beta(r-j+1,j-1)} \left(1-z\right)^{j-2}  z^{r-j} \log(u z)  dz\\
&=&  \frac{\e^{-u/2} u^{r-1}\beta(r-j+1,j-1)}{ 2^{r-1} \Gamma(r-j+1)\Gamma(j-1)} \left[\log(u)\int_{z=0}^{1} \frac{1}{\beta(r-j+1,j-1)} \left(1-z\right)^{j-2}  z^{r-j} dz\right.\\
&& \quad \quad \left. +\int_{z=0}^{1} \frac{1}{\beta(r-j+1,j-1)} \left(1-z\right)^{j-2}  z^{r-j} \log(z)  dz\right]
\end{eqnarray*}
where $\beta(r-j+1,j-1) = \frac{\Gamma(r-j+1) \Gamma(j-1)}{\Gamma(r)}$ is the Beta function. Note that while $\log(z)\to-\infty$ as $z\to 0$, the integral converges as this is the mean of $\log(Z)$ when $Z$ has a Beta distribution. Accordingly, we have
$$I_2(x) > \frac{\e^{-u/2} u^{r-1}}{2^{r-1} \Gamma(r)}\left[\log (u) + \psi(r-j+1) - \psi(r) \right].$$
To complete the proof, remember that $Pr(\bx +\by>u) = \frac{\Gamma(r,u/2)}{\Gamma(r)} $. Thus, we have
\begin{multline*}
    \E[\log(1+x)|\bx +\by > u]
    > \frac{\Gamma(r)}{\Gamma(r,u/2)}\Bigg(\frac{\e^{-u/2} u^{r-1}}{2^{r-1} \Gamma(r)} \left[\log (u) + \psi(r-j+1) - \psi(r) \right]
\\
+ \frac{\Gamma(r-j+1,u/2)}{\Gamma(r-j+1)} \log(1+u) +e\frac{\Gamma(r-j,1+u/2)}{\Gamma(r-j+1)} \Bigg).
    \end{multline*}

\bibliographystyle{IEEEtran}
\bibliography{bibliography}

\begin{thebibliography}{10}
\providecommand{\url}[1]{#1}
\csname url@samestyle\endcsname
\providecommand{\newblock}{\relax}
\providecommand{\bibinfo}[2]{#2}
\providecommand{\BIBentrySTDinterwordspacing}{\spaceskip=0pt\relax}
\providecommand{\BIBentryALTinterwordstretchfactor}{4}
\providecommand{\BIBentryALTinterwordspacing}{\spaceskip=\fontdimen2\font plus
\BIBentryALTinterwordstretchfactor\fontdimen3\font minus
  \fontdimen4\font\relax}
\providecommand{\BIBforeignlanguage}[2]{{%
\expandafter\ifx\csname l@#1\endcsname\relax
\typeout{** WARNING: IEEEtran.bst: No hyphenation pattern has been}%
\typeout{** loaded for the language `#1'. Using the pattern for}%
\typeout{** the default language instead.}%
\else
\language=\csname l@#1\endcsname
\fi
#2}}
\providecommand{\BIBdecl}{\relax}
\BIBdecl

\bibitem{yoo2006optimality}
T.~Yoo and A.~Goldsmith, ``On the optimality of multiantenna broadcast
  scheduling using zero-forcing beamforming,'' \emph{Selected Areas in
  Communications, IEEE Journal on}, vol.~24, no.~3, pp. 528--541, 2006.

\bibitem{jagannathan2007scheduling}
K.~Jagannathan, S.~Borst, P.~Whiting, and E.~Modiano, ``Scheduling of
  multi-antenna broadcast systems with heterogeneous users,'' \emph{Selected
  Areas in Communications, IEEE Journal on}, vol.~25, no.~7, pp. 1424--1434,
  2007.

\bibitem{swannack2004low}
C.~Swannack, E.~Uysal-Biyikoglu, and G.~Wornell, ``Low complexity multiuser
  scheduling for maximizing throughput in the {MIMO} broadcast channel,'' in
  \emph{Proc. Allerton Conf. Communications, Control and Computing}, 2004.

\bibitem{kim2005scheduling}
J.~Kim, S.~Park, J.~Lee, J.~Lee, and H.~Jung, ``A scheduling algorithm combined
  with zero-forcing beamforming for a multiuser {MIMO} wireless system,'' in
  \emph{IEEE Vehicular Technology Conference}, vol.~1, 2005, pp. 211--215.

\bibitem{airy2003spatially}
M.~Airy, S.~Shakkattai, and R.~Heath~Jr, ``Spatially greedy scheduling in
  multi-user {MIMO} wireless systems,'' in \emph{the Thirty-Seventh Asilomar
  Conference on Signals, Systems and Computers}, vol.~1.\hskip 1em plus 0.5em
  minus 0.4em\relax IEEE, 2003, pp. 982--986.

\bibitem{knopp1995information}
R.~Knopp and P.~Humblet, ``Information capacity and power control in
  single-cell multiuser communications,'' in \emph{IEEE International
  Conference on Communications, Seattle}, vol.~1, 1995, pp. 331--335.

\bibitem{gozali2003impact}
R.~Gozali, R.~Buehrer, and B.~Woerner, ``The impact of multiuser diversity on
  space-time block coding,'' \emph{Communications Letters, IEEE}, vol.~7,
  no.~5, pp. 213--215, 2003.

\bibitem{qin2003exploiting}
X.~Qin and R.~Berry, ``Exploiting multiuser diversity for medium access control
  in wireless networks,'' in \emph{INFOCOM 2003. Twenty-Second Annual Joint
  Conference of the IEEE Computer and Communications. IEEE Societies},
  vol.~2.\hskip 1em plus 0.5em minus 0.4em\relax IEEE, 2003, pp. 1084--1094.

\bibitem{qin2006distributed}
------, ``Distributed approaches for exploiting multiuser diversity in wireless
  networks,'' \emph{Information Theory, IEEE Transactions on}, vol.~52, no.~2,
  pp. 392--413, 2006.

\bibitem{gesbert2004much}
D.~Gesbert and M.-S. Alouini, ``How much feedback is multi-user diversity
  really worth?'' in \emph{Communications, 2004 IEEE International Conference
  on}, vol.~1.\hskip 1em plus 0.5em minus 0.4em\relax IEEE, 2004, pp. 234--238.

\bibitem{kampeas2012opportunistic}
J.~Kampeas, A.~Cohen, and O.~Gurewitz, ``Capacity of distributed opportunistic
  scheduling in nonhomogeneous networks,'' \emph{Information Theory, IEEE
  Transactions on}, vol.~60, no.~11, pp. 7231--7247, 2014.

\bibitem{kampeas2012capdisthetnet}
------, ``Capacity of distributed opportunistic scheduling in heterogeneous
  networks,'' in \emph{Proceedings of the Annual Allerton Conference on
  Communication Control and Computing}, 2012.

\bibitem{bai2006opportunistic}
K.~Bai and J.~Zhang, ``Opportunistic multichannel aloha: distributed
  multiaccess control scheme for {OFDMA} wireless networks,'' \emph{Vehicular
  Technology, IEEE Transactions on}, vol.~55, no.~3, pp. 848--855, 2006.

\bibitem{qin2008distributed}
X.~Qin and R.~Berry, ``Distributed power allocation and scheduling for parallel
  channel wireless networks,'' \emph{Wireless Networks}, vol.~14, no.~5, pp.
  601--613, 2008.

\bibitem{shen2006low}
Z.~Shen, R.~Chen, J.~Andrews, R.~Heath, and B.~Evans, ``Low complexity user
  selection algorithms for multiuser {MIMO} systems with block
  diagonalization,'' \emph{Signal Processing, IEEE Transactions on}, vol.~54,
  no.~9, pp. 3658--3663, 2006.

\bibitem{jagannathan2006efficient}
K.~Jagannathan, S.~Borst, P.~Whiting, and E.~Modiano, ``Efficient scheduling of
  multi-user multi-antenna systems,'' in \emph{Modeling and Optimization in
  Mobile, Ad Hoc and Wireless Networks, 2006 4th International Symposium
  on}.\hskip 1em plus 0.5em minus 0.4em\relax IEEE, 2006, pp. 1--8.

\bibitem{weingarten2006capacity}
H.~Weingarten, Y.~Steinberg, and S.~Shamai, ``The capacity region of the
  gaussian multiple-input multiple-output broadcast channel,''
  \emph{Information Theory, IEEE Transactions on}, vol.~52, no.~9, pp.
  3936--3964, 2006.

\bibitem{primolevo2005channel}
G.~Primolevo, O.~Simeone, and U.~Spagnolini, ``Channel aware scheduling for
  broadcast {MIMO} systems with orthogonal linear precoding and fairness
  constraints,'' in \emph{IEEE International Conference on Communications},
  vol.~4, 2005, pp. 2749--2753.

\bibitem{yoo2006finite}
T.~Yoo, N.~Jindal, and A.~Goldsmith, ``Finite-rate feedback {MIMO} broadcast
  channels with a large number of users,'' in \emph{Information Theory, 2006
  IEEE International Symposium on}.\hskip 1em plus 0.5em minus 0.4em\relax
  IEEE, 2006, pp. 1214--1218.

\bibitem{gartner2006multiuser}
M.~Gartner and H.~Bolcskei, ``Multiuser space-time/frequency code design,'' in
  \emph{Information Theory, 2006 IEEE International Symposium on}.\hskip 1em
  plus 0.5em minus 0.4em\relax IEEE, 2006, pp. 2819--2823.

\bibitem{pun2007opportunistic}
M.~Pun, V.~Koivunen, and H.~Poor, ``Opportunistic scheduling and beamforming
  for {MIMO-SDMA} downlink systems with linear combining,'' in \emph{IEEE 18th
  International Symposium on Personal, Indoor and Mobile Radio Communications},
  2007, pp. 1--6.

\bibitem{wang2007coverage}
L.~Wang, C.~Chiu, C.~Yeh, and C.~Li, ``Coverage enhancement for {OFDM}-based
  spatial multiplexing systems by scheduling,'' in \emph{IEEE Wireless
  Communications and Networking Conference, WCNC}.\hskip 1em plus 0.5em minus
  0.4em\relax IEEE, 2007, pp. 1439--1443.

\bibitem{zakhour2011minmax}
R.~Zakhour and S.~Hanly, ``Min-max fair coordinated beamforming via large
  system analysis,'' \emph{IEEE International Symposium on Information Theory
  Proceedings}, pp. 1896--1900, 2011.

\bibitem{choi2008capacity}
W.~Choi and J.~Andrews, ``The capacity gain from intercell scheduling in
  multi-antenna systems,'' \emph{Wireless Communications, IEEE Transactions
  on}, vol.~7, no.~2, pp. 714--725, 2008.

\bibitem{caire2006mimo}
G.~Caire, ``{MIMO} downlink joint processing and scheduling: a survey of
  classical and recent results,'' in \emph{Proc. Workshop on Information Theory
  and Its Applications}, 2006.

\bibitem{hassibi2007fundamental}
B.~Hassibi and M.~Sharif, ``Fundamental limits in {MIMO} broadcast channels,''
  \emph{Selected Areas in Communications, IEEE Journal on}, vol.~25, no.~7, pp.
  1333--1344, 2007.

\bibitem{sharif2007comparison}
M.~Sharif and B.~Hassibi, ``A comparison of time-sharing, dpc, and beamforming
  for {MIMO} broadcast channels with many users,'' \emph{Communications, IEEE
  Transactions on}, vol.~55, no.~1, pp. 11--15, 2007.

\bibitem{bayesteh2008user}
A.~Bayesteh and A.~K. Khandani, ``On the user selection for {MIMO} broadcast
  channels,'' \emph{Information Theory, IEEE Transactions on}, vol.~54, no.~3,
  pp. 1086--1107, 2008.

\bibitem{sharif2005capacity}
M.~Sharif and B.~Hassibi, ``On the capacity of {MIMO} broadcast channels with
  partial side information,'' \emph{Information Theory, IEEE Transactions on},
  vol.~51, no.~2, pp. 506--522, 2005.

\bibitem{hochwald2004multiple}
B.~M. Hochwald, T.~L. Marzetta, and V.~Tarokh, ``Multiple-antenna channel
  hardening and its implications for rate feedback and scheduling,''
  \emph{Information Theory, IEEE Transactions on}, vol.~50, no.~9, pp.
  1893--1909, 2004.

\bibitem{caire1999capacity}
G.~Caire and S.~Shamai, ``On the capacity of some channels with channel state
  information,'' \emph{Information Theory, IEEE Transactions on}, vol.~45,
  no.~6, pp. 2007--2019, 1999.

\bibitem{ozarow1994information}
L.~H. Ozarow, S.~Shamai, and A.~D. Wyner, ``Information theoretic
  considerations for cellular mobile radio,'' \emph{IEEE transactions on
  Vehicular Technology}, vol.~43, no.~2, pp. 359--378, 1994.

\bibitem{de2006extreme}
L.~De~Haan and A.~Ferreira, \emph{Extreme value theory: an introduction}.\hskip
  1em plus 0.5em minus 0.4em\relax Springer Verlag, 2006.

\bibitem{leadbetter1983}
M.~Leadbetter, \emph{Extremes and Related Properties of Random Sequences and
  Processes}.\hskip 1em plus 0.5em minus 0.4em\relax Springer-Verlag, N.Y,
  1983.

\bibitem{embrechts2011modelling}
P.~Embrechts, C.~Kl{\"u}ppelberg, and T.~Mikosch, \emph{Modelling extremal
  events: for insurance and finance}.\hskip 1em plus 0.5em minus 0.4em\relax
  Springer, 2011, vol.~33.

\bibitem{tse2005fundamentals}
D.~Tse and P.~Viswanath, \emph{Fundamentals of wireless communication}.\hskip
  1em plus 0.5em minus 0.4em\relax Cambridge university press, 2005.

\bibitem{li2006distribution}
P.~Li, D.~Paul, R.~Narasimhan, and J.~Cioffi, ``On the distribution of {SINR}
  for the {MMSE MIMO} receiver and performance analysis,'' \emph{Information
  Theory, IEEE Transactions on}, vol.~52, no.~1, pp. 271--286, 2006.

\bibitem{kim2008performance}
N.~Kim, Y.~Lee, and H.~Park, ``Performance analysis of {MIMO} system with
  linear mmse receiver,'' \emph{Wireless Communications, IEEE Transactions on},
  vol.~7, no.~11, pp. 4474--4478, 2008.

\bibitem{xu2015analysis}
P.~Xu, J.~Wang, J.~Wang, and F.~Qi, ``Analysis and design of channel estimation
  in multicell multiuser {MIMO OFDM} systems,'' \emph{Vehicular Technology,
  IEEE Transactions on}, vol.~64, no.~2, pp. 610--620, 2015.

\bibitem{suh2003preamble}
C.~Suh, C.-S. Hwang, and H.~Choi, ``Preamble design for channel estimation in
  {MIMO-OFDM} systems,'' in \emph{Global Telecommunications Conference, 2003.
  GLOBECOM'03. IEEE}, vol.~1.\hskip 1em plus 0.5em minus 0.4em\relax IEEE,
  2003, pp. 317--321.

\bibitem{magistretti2014802}
E.~Magistretti, O.~Gurewitz, and E.~W. Knightly, ``802.11 ec: collision
  avoidance without control messages,'' \emph{IEEE/ACM Transactions on
  Networking {(TON)}}, vol.~22, no.~6, pp. 1845--1858, 2014.

\bibitem{janssen2008gaussian}
A.~Janssen, J.~Van~Leeuwaarden, and B.~Zwart, ``Gaussian expansions and bounds
  for the poisson distribution applied to the erlang b formula,''
  \emph{Advances in Applied Probability}, vol.~40, no.~1, pp. 122--143, 2008.

\bibitem{lau2004design}
V.~Lau, Y.~Liu, and T.-A. Chen, ``On the design of {MIMO} block-fading channels
  with feedback-link capacity constraint,'' \emph{IEEE Transactions on
  Communications}, vol.~52, no.~1, pp. 62--70, 2004.

\bibitem{bai1996bounds}
Z.~Bai and G.~H. Golub, ``Bounds for the trace of the inverse and the
  determinant of symmetric positive definite matrices,'' \emph{Annals of
  Numerical Mathematics}, vol.~4, pp. 29--38, 1996.

\bibitem{GammaBound}
\BIBentryALTinterwordspacing
``Are there well known lower bounds for the upper incomplete gamma function?''
  [Online]. Available: \url{http://math.stackexchange.com/questions/129170}
\BIBentrySTDinterwordspacing

\bibitem{smith1989extreme}
R.~Smith, ``Extreme value analysis of environmental time series: an application
  to trend detection in ground-level ozone,'' \emph{Statistical Science},
  vol.~4, no.~4, pp. 367--377, 1989.

\bibitem{kallenberg1983random}
O.~Kallenberg, \emph{Random Measures}, 3rd~ed.\hskip 1em plus 0.5em minus
  0.4em\relax Akademie Verlag, Berlin, Academic Press, New York, 1983.

\bibitem{durrett2010probability}
R.~Durrett, \emph{Probability: theory and examples}.\hskip 1em plus 0.5em minus
  0.4em\relax Cambridge university press, 2010.

\bibitem{coles2001introduction}
S.~Coles, \emph{An introduction to statistical modeling of extreme
  values}.\hskip 1em plus 0.5em minus 0.4em\relax Springer Verlag, 2001.

\bibitem{telatar1999capacity}
E.~Telatar, ``Capacity of multi-antenna gaussian channels,'' \emph{European
  transactions on telecommunications}, vol.~10, no.~6, pp. 585--595, 1999.

\bibitem{tulino2004random}
A.~M. Tulino and S.~Verd{\'u}, \emph{Random matrix theory and wireless
  communications}.\hskip 1em plus 0.5em minus 0.4em\relax Now Publishers Inc,
  2004, vol.~1.

\end{thebibliography}
\end{document}